\documentclass[lettersize,journal]{IEEEtran} 
\usepackage{amsmath,amsfonts}

\makeatletter
\DeclareRobustCommand\widecheck[1]{{\mathpalette\@widecheck{#1}}}
\def\@widecheck#1#2{%
	\setbox\z@\hbox{\m@th$#1#2$}%
	\setbox\tw@\hbox{\m@th$#1%
		\widehat{%
			\vrule\@width\z@\@height\ht\z@
			\vrule\@height\z@\@width\wd\z@}$}%
	\dp\tw@-\ht\z@
	\@tempdima\ht\z@ \advance\@tempdima2\ht\tw@ \divide\@tempdima\thr@@
	\setbox\tw@\hbox{%
		\raise\@tempdima\hbox{\scalebox{1}[-1]{\lower\@tempdima\box
				\tw@}}}%
	{\ooalign{\box\tw@ \cr \box\z@}}}
\makeatother

\makeatletter
\newcommand*\rel@kern[1]{\kern#1\dimexpr\macc@kerna}
\newcommand*\widebar[1]{%
	\begingroup
	\def\mathaccent##1##2{%
		\rel@kern{0.8}%
		\overline{\rel@kern{-0.8}\macc@nucleus\rel@kern{0.2}}%
		\rel@kern{-0.2}%
	}%
	\macc@depth\@ne
	\let\math@bgroup\@empty \let\math@egroup\macc@set@skewchar
	\mathsurround\z@ \frozen@everymath{\mathgroup\macc@group\relax}%
	\macc@set@skewchar\relax
	\let\mathaccentV\macc@nested@a
	\macc@nested@a\relax111{#1}%
	\endgroup
}
\makeatother

\usepackage{subfigure}
\usepackage{algorithm}
\usepackage{algpseudocode}
\usepackage{array}
\usepackage{stfloats}
\usepackage{url}
\usepackage{verbatim}
\usepackage{graphicx}
\usepackage{cite}
\usepackage{xcolor}
\usepackage{multirow}
\usepackage{amssymb}
\usepackage{booktabs}

\DeclareMathOperator*{\argmax}{arg\,max}
\DeclareMathOperator*{\argmin}{arg\,min}
\DeclareMathOperator*{\vect}{vec}
\DeclareMathOperator*{\invect}{invec}

\DeclareMathOperator*{\blkdiag}{blkdiag}

\usepackage{amsthm}

\newtheorem{lemma}{Lemma}

\newcommand\norm[1]{\lVert#1\rVert}
\newcommand\normF[1]{\lVert#1\rVert_{\mathrm{F}}}

\newcommand{\ut}{_{\mathrm{t}}}
\newcommand{\ur}{_{\mathrm{r}}}

\newcommand{\los}{_{\mathrm{LoS}}}
\newcommand{\nlos}{_{\mathrm{NLoS}}}
\newcommand{\nlosp}{_{\mathrm{NLoS}, \mathcal{P}}}
\newcommand{\nlospp}{_{\mathrm{NLoS}, \mathcal{P}'}}
\newcommand{\A}{\mathbf{A}}
\newcommand{\At}{A_{\mathrm{t}}}
\newcommand{\Ar}{A_{\mathrm{r}}}
\newcommand{\Ats}{A_{\mathrm{t}, \mathrm{s}}}
\newcommand{\Ars}{A_{\mathrm{r}, \mathrm{s}}}
\newcommand{\B}{\mathbf{B}}
\newcommand{\C}{\mathbf{C}}

\newcommand{\K}{\mathbf{K}}

\newcommand{\Nt}{N_{\mathrm{t}}}
\newcommand{\Nr}{N_{\mathrm{r}}}
\newcommand{\Nts}{N_{\mathrm{t}, \mathrm{s}}}
\newcommand{\Nrs}{N_{\mathrm{r}, \mathrm{s}}}

\newcommand{\Mt}{M_{\mathrm{t}}}
\newcommand{\Mr}{M_{\mathrm{r}}}
\newcommand{\Mts}{M_{\mathrm{t}, \mathrm{s}}}
\newcommand{\Mrs}{M_{\mathrm{r}, \mathrm{s}}}

\newcommand{\Kt}{K_{\mathrm{t}}}
\newcommand{\Kr}{K_{\mathrm{r}}}

\newcommand{\Stj}{\mathbf{S}_{\mathrm{t}, j}}
\newcommand{\Sri}{\mathbf{S}_{\mathrm{r}, i}}

\newcommand{\deltarm}{\delta_{\mathrm{r}, m}}
\newcommand{\deltatn}{\delta_{\mathrm{t}, n}}

\newcommand{\nuri}{\nu_{\mathrm{r}, i}}
\newcommand{\nutj}{\nu_{\mathrm{t}, j}}

\newcommand{\F}{\mathbf{F}}

\newcommand{\W}{\mathbf{W}}

\newcommand{\I}{\mathbf{I}}
\newcommand{\Y}{\mathbf{Y}}
\newcommand{\y}{\mathbf{y}}

\newcommand{\Q}{\mathcal{Q}}

\newcommand{\thetat}{\theta_{\mathrm{t}}}
\newcommand{\thetar}{\theta_{\mathrm{r}}}

\newcommand{\rhotl}{\rho_{\mathrm{t}, l}}
\newcommand{\rhorl}{\rho_{\mathrm{r}, l}}

\newcommand{\varrhotl}{\varrho_{\mathrm{t}, l}}
\newcommand{\varrhorl}{\varrho_{\mathrm{r}, l}}

\newcommand{\gammatij}{\gamma_{\mathrm{t}, i, j}}
\newcommand{\gammarij}{\gamma_{\mathrm{r}, i, j}}

\newcommand{\xiti}{\xi_{\mathrm{t}, i}}
\newcommand{\xirj}{\xi_{\mathrm{r}, j}}

\newcommand{\phir}{\phi_{\mathrm{r}}}

\newcommand{\at}{\mathbf{a}_{\mathrm{t}}}
\newcommand{\ar}{\mathbf{a}_{\mathrm{r}}}

\newcommand{\Dt}{\mathbf{D}_{\mathrm{t}}}
\newcommand{\Dr}{\mathbf{D}_{\mathrm{r}}}

\newcommand{\QDt}{Q_{\mathbf{D}_{\mathrm{t}}}}
\newcommand{\QDr}{Q_{\mathbf{D}_{\mathrm{r}}}}

\newcommand{\varphit}{\varphi_{\mathrm{t}}}

\newcommand{\varphir}{\varphi_{\mathrm{r}}}

\newcommand{\alphat}{\alpha_{\mathrm{t}}}
\newcommand{\alphar}{\alpha_{\mathrm{r}}}

\newcommand{\checkxiti}{\check{\xi}_{\mathrm{t}, i}}
\newcommand{\checkxirj}{\check{\xi}_{\mathrm{r}, j}}

\newcommand{\hatvarphit}{\hat{\varphi}_{\mathrm{t}}}
\newcommand{\hatvarphir}{\hat{\varphi}_{\mathrm{r}}}

\newcommand{\hateta}{\hat{\eta}}

\newcommand{\hatOmegar}{\widehat{\mathbf{\Omega}}_{\mathrm{r}}}
\newcommand{\hatOmegat}{\widehat{\mathbf{\Omega}}_{\mathrm{t}}}

\newcommand{\hatDr}{\widehat{\mathbf{D}}_{\mathrm{r}}}
\newcommand{\hatDt}{\widehat{\mathbf{D}}_{\mathrm{t}}}

\newcommand{\hatXit}{\hat{\boldsymbol{\xi}}_{\mathrm{t}}}
\newcommand{\hatXir}{\hat{\boldsymbol{\xi}}_{\mathrm{r}}}

\newcommand{\hatalphat}{\hat{\alpha}_{\mathrm{t}}}
\newcommand{\hatalphar}{\hat{\alpha}_{\mathrm{r}}}

\newcommand{\hatxiti}{\hat{\xi}_{\mathrm{t}, i}}
\newcommand{\hatxirj}{\hat{\xi}_{\mathrm{r}, j}}

\newcommand{\barthetat}{\bar{\theta}_{\mathrm{t}}}
\newcommand{\barthetar}{\bar{\theta}_{\mathrm{r}}}
\newcommand{\barphir}{\bar{\phi}_{\mathrm{r}}}
\newcommand{\barR}{\bar{R}}

\newcommand{\bareta}{\bar{\eta}}

\newcommand{\barXit}{\bar{\boldsymbol{\xi}}_{\mathrm{t}}}
\newcommand{\barXir}{\bar{\boldsymbol{\xi}}_{\mathrm{r}}}

\newcommand{\barxiti}{\bar{\xi}_{\mathrm{t}, i}}
\newcommand{\barxirj}{\bar{\xi}_{\mathrm{r}, j}}

\newcommand{\barvarphit}{\bar{\varphi}_{\mathrm{t}}}
\newcommand{\barvarphir}{\bar{\varphi}_{\mathrm{r}}}

\newcommand{\baralphat}{\bar{\alpha}_{\mathrm{t}}}
\newcommand{\baralphar}{\bar{\alpha}_{\mathrm{r}}}

\newcommand{\baratj}{\bar{\mathbf{a}}_{\mathrm{t}, j}}
\newcommand{\barari}{\bar{\mathbf{a}}_{\mathrm{r}, i}}

\newcommand{\tildeatj}{\tilde{\mathbf{a}}_{\mathrm{t}, j}}
\newcommand{\tildeari}{\tilde{\mathbf{a}}_{\mathrm{r}, i}}

\begin{document}
	
	\title{Low-Complexity On-Grid Channel Estimation for Partially-Connected Hybrid XL-MIMO}
	
	\author{
		\IEEEauthorblockN{Sunho Kim},~\IEEEmembership{Graduate Student Member,~IEEE},
		and
		\IEEEauthorblockN{Wan Choi},~\IEEEmembership{Fellow,~IEEE}
		\vspace{-0.1in}
		\thanks{S.~Kim and W.~Choi are with the Department of Electrical and Computer Engineering, and the institute of New Media and Communications and   Seoul National University
			(SNU), Seoul 08826, Korea (e-mail: \{sunhokim,~wanchoi\}@snu.ac.kr).  (\emph{Corresponding Author: Wan Choi})} 
	\vspace{-0.1in}
}

\markboth{ draft }%
{Shell \MakeLowercase{\textit{et al.}}: Low-Complexity Channel Estimation for Near-Field XL-MIMO}

\maketitle

\begin{abstract}
	This paper addresses the challenge of channel estimation in extremely large-scale multiple-input multiple-output (XL-MIMO) systems, pivotal for the advancement of 6G communications. 
	XL-MIMO systems, characterized by their vast antenna arrays, necessitate accurate channel state information (CSI) to leverage high spatial multiplexing and beamforming gains. 
	However, conventional channel estimation methods for near-field XL-MIMO encounter significant computational complexity due to the exceedingly high parameter quantization levels needed for estimating the parametric near-field channel.
	To address this, we propose a low-complexity two-stage on-grid channel estimation algorithm designed for near-field XL-MIMO systems.
	The first stage focuses on estimating the LoS channel component while treating the NLoS paths as interference. This estimation is accomplished through an alternating subarray-wise array gain maximization (ASAGM) approach based on the piecewise outer product model (SOPM).	In the second stage, we estimate the NLoS channel component by utilizing the sensing matrix refinement-based orthogonal matching pursuit (SMR-OMP) algorithm. This approach helps reduce the high computational complexity associated with large-dimensional joint sensing matrices. 
	Simulation results demonstrate the effectiveness of our proposed low-complexity method, showcasing its significant superiority over existing near-field XL-MIMO channel estimation techniques, particularly in intermediate and high SNR regimes, and in practical scenarios involving arbitrary array placements.

\end{abstract}


\IEEEpeerreviewmaketitle

\section{Introduction}\label{Introduction}
Multiple-input multiple-output (MIMO) is a pivotal technology that significantly enhances the performance of 5G, offering substantial gains through high spatial multiplexing and beamforming \cite{andrews2014will}.
However, conventional MIMO configurations may fall short of meeting the stringent performance demands expected with the advent of 6G \cite{viswanathan2020communications}. To address these challenges, recent advancements have led to the development of extremely large-scale MIMO (XL-MIMO), which utilizes much larger antenna arrays compared to traditional MIMO systems, specifically designed to meet these evolving requirements \cite{wang2023extremely}. Moreover, the availability of abundant bandwidth at high-frequency ranges has accelerated the adoption of mmWave and sub-terahertz bands \cite{tataria20216g}. As a result, high-frequency XL-MIMO communication techniques are emerging as a critical enabler for 6G communications.
 
Due to the significant power consumption and high hardware costs, a fully digital architecture that assigns a dedicated radio frequency (RF) chain to each antenna is impractical for XL-MIMO systems. As a result, a hybrid MIMO architecture has become a popular choice for XL-MIMO, connecting a limited number of RF chains to a much larger array of antennas \cite{alkhateeb2014channel, bogale2016number}. Among hybrid MIMO configurations, the partially-connected architecture, where each RF chain is linked to only a subset of antennas, has garnered increasing attention for its reduced hardware complexity \cite{gao2016energy, majidzadeh2017hybrid}. Related research indicates that this architecture can achieve improved power efficiency with minimal performance trade-offs and, in some cases, even enhanced performance when considering factors such as power dissipation and nonlinear distortion \cite{du2018hybrid, song2019fully}. Thus, the partially-connected hybrid architecture is gaining recognition as a practical and efficient solution for implementing XL-MIMO systems.

To harness the spatial multiplexing and beamforming gains facilitated by MIMO architecture, the acquisition of precise channel state information (CSI) is essential.
Substantial research has focused on MIMO channel estimation, predominantly based on the far-field channel model \cite{lee2016channel, rodriguez2018frequency}.
However, as the antenna aperture area increases, the planar wavefront assumption inherent in the far-field channel model becomes invalid, making these conventional methods unsuitable for XL-MIMO. 
To address this challenge, recent advancements in channel estimation literature have introduced channel model based on spherical wave propagation, i.e., near-field channel. 

The spherical wave model can be divided into two types: the non-uniform spherical wave model (NUSWM) \cite{zhou2015spherical} and the uniform spherical wave model (USWM) \cite{le2019massive}. As the array dimension increases further, the NUSWM provides a more accurate representation by accounting for variations in the received signal power across the antenna elements. On the other hand, in the USWM, it is assumed that the received signal power is uniform across all antenna elements to simplify the channel model.
While the USWM is more analytically tractable than the NUSWM, its complex phase term still presents challenges in model analysis. To address this, Fresnel approximation with uniform power assumption \cite{goodman2005introduction} is widely employed for near-field MIMO systems \cite{bohagen2005construction, do2020reconfigurable, do2023parabolic, xi2023gridless}. 
Following the terminology in \cite{do2023parabolic}, we refer to the wavefront model derived from the Fresnel approximation as the parabolic wavefront model.

Recent studies on near-field XL-MIMO channel estimation can be classified into two categories based on antenna configurations: those where either the transmitter or receiver is equipped with a single antenna, referred to as multiple-input single-output (MISO), and those where both the transmitter and receiver are equipped with multiple antennas, referred to as MIMO.
For MISO systems, the Fraunhofer distance (FD) \cite{selvan2017fraunhofer} is commonly used as a criterion to distinguish between the far-field and the parabolic wavefront model-based near-field regions of the MISO channel. For example, in a MISO system operating at a carrier frequency of 60 GHz, where the transmitter is equipped with a uniform linear array (ULA) comprising 256 antennas with half-wavelength spacing, the FD is approximately 163 meters.
Consequently, the near-field region encompasses a significant portion of conventional cell coverage areas. This has spurred extensive research into near-field channel estimation methods based on the USWM or the parabolic wavefront model for MISO systems \cite{han2020channel, cui2022channel, zhang2023near, kang2024pilot} with the assumption that the communication distance is sufficiently large to disregard variations in received signal power.

The authors of \cite{han2020channel} proposed a channel estimation method for near-field MISO systems, utilizing a subarray-wise approach to mitigate the computational complexity associated with large-dimensional near-field array steering vectors. In \cite{cui2022channel}, a polar domain dictionary was developed for the sparse representation of the near-field channel, facilitating the use of a compressive sensing (CS) approach to efficiently estimate channel parameters and gains. Building on this, \cite{zhang2023near} introduced a joint dictionary learning and sparse recovery method based on a distance-parameterized angular domain dictionary, which addresses the dimensionality challenges of the polar domain dictionary introduced in \cite{cui2022channel}. 
Additionally, \cite{kang2024pilot} proposed a pilot signal design method aimed at minimizing the mutual coherence of the sensing matrix based on the polar domain dictionary and Bayesian matching pursuit for hybrid near and far-field channels. 

While numerous studies have focused on near-field MISO channel estimation, it is important to note that MISO configurations are unable to fully exploit the spatial multiplexing and beamforming advantages inherent in MIMO architecture. However, research on near-field MIMO channel estimation remains limited, primarily due to the high dimensionality and complexity of the channel model.
Specifically, while a non-line-of-sight (NLoS) channel can be accurately represented as the outer product of the existing near-field array steering vectors developed for MISO systems, this approach does not directly extend to LoS channels.
To address this, \cite{lu2023near} proposed a two-stage channel estimation method where the geometric parameters of the NUSWM-based LoS channel is estimated under the maximum likelihood criterion, followed by NLoS channel estimation using the CS technique. However, the channel model considered in \cite{lu2023near} is limited to scenarios where the transmit and receive antenna arrays lie in the same plane, which undermines its practicality. Furthermore, the estimation algorithm in \cite{lu2023near} suffers from prohibitively high computational complexity due to its excessively large parameter quantization level for the LoS channel and the dictionaries representing the near-field array steering vectors for the NLoS channel. 
On the other hand, \cite{shi2024double} proposed a unified LoS and NLoS sparse channel representation, employing parabolic wavefront model and polar domain dictionaries for CS-based estimation. However, this approach faces challenges due to the excessively high dimensionality of the polar domain dictionary in XL-MIMO systems. Additionally, the inclusion of extra parameters, driven by the intricate nature of the LoS channel model, further increases the dimension of the dictionary.
In contrast, \cite{tarboush2024cross} proposed a low-complexity subarray-wise channel estimation method based on the CS technique for partially-connected MIMO architecture. However, their approach does not account for the relationships between the channel submatrices corresponding to different subarrays, resulting in the compromised estimation accuracy.

To develop an accurate yet low-complexity channel estimation method for practical near-field partially-connected hybrid XL-MIMO, we introduce the subarray-wise outer product model (SOPM), which enables an accurate outer product representation of the channel matrix between each transmit and receive subarray. We first analyze the validity range of the SOPM—referred to as the subarray-wise outer product distance (SOPD)—by examining the phase error of the SOPM compared to the USWM. Building on this foundation, we propose a two-stage on-grid channel estimation algorithm. Specifically, to leverage the distinct characteristics of the LoS and NLoS channel models, we first estimate the LoS channel while treating the NLoS channel as interference. Afterward, we estimate the NLoS channel by removing the influence of the previously estimated LoS component. The contributions of this paper are summarized as follows.
\begin{itemize}
	\item  We introduce a new SOPM to approximate the LoS channel matrix. In this model, the complicated phase terms of the parabolic wavefront mode, 
which depend on both the transmit and receive antenna indices, are approximated with reasonable accuracy as a sum of terms that depend on either the transmit or receive antenna index separately. This simplification allows us to represent the channel submatrices as the well-studied outer product of array steering vectors, enabling the application of simple signal processing methods. Furthermore, we propose the subarray-wise outer product distance (SOPD) to identify the regions where the LoS channel can be accurately modeled using the SOPM.	
	\item   Building on the structured nature of SOPM,  we develop a low-complexity LoS channel estimation algorithm tailored for near-field partially-connected XL-MIMO systems. Rather than relying on computationally prohibitive high-dimensional parameter searches, we introduce the alternating subarray-wise array gain maximization (ASAGM) framework. This  method decomposes the challenging estimation task into multiple low-dimensional subproblems based on SOPM, enabling an efficient iterative approach that maximizes the correlation between received pilot signals and array steering vectors. This reduction in parameter dimensionality not only alleviates computational complexity but also enables the use of higher parameter quantization levels, resulting in enhanced estimation performance.  
	\item  We propose a sensing matrix refinement-based orthogonal matching pursuit (SMR-OMP) algorithm for estimating the NLoS channel with dramatically reduces dictionary quantization complexity.  After isolating and removing the LoS component using SOPM, we leverage a two-step refinement strategy: (i) using an enhanced simultaneous OMP (SOMP) algorithm to coarsely detect support sets for transmit and receive dictionaries separately and (ii) constructing a refined joint dictionary based on these detected supports. This approach mitigates the computational burden of high-dimensional joint dictionary searches, significantly improving both efficiency and accuracy.	
	\item Extensive simulations confirm that our proposed framework, underpinned by the novel SOPM channel model, substantially outperforms existing near-field XL-MIMO channel estimation techniques \cite{lu2023near, shi2024double} and the polar-domain dictionary-based OMP method \cite{cui2022channel}. Our method achieves superior estimation accuracy in intermediate and high SNR regimes, while maintaining drastically lower computational complexity. These gains are particularly pronounced in real-world deployments with arbitrary array placements.
\end{itemize}

The remainder of this paper is organized as follows. 
Section \ref{System and Channel Models} introduces the preliminaries, including the system model and the LoS channel models based on the spherical wavefront, the parabolic wavefront, and the SOPM, as well as the NLoS channel model.
Section \ref{Proposed Two-stage Channel Estimation} details the proposed two-stage channel estimation method, which comprises ASAGM for LoS channel estimation and SMR-OMP for NLoS channel estimation. 
Section \ref{Simulation Results} provides simulation results, followed by the conclusion in Section \ref{Conclusions}.

\textit{Notations}:
Boldface uppercase, boldface lowercase, and normal face lowercase letters represent matrices, vectors, and scalars, respectively.
The transpose, conjugate, and conjugate transpose of $\mathbf{X}$ are represented by $\mathbf{X}^T$, $\mathbf{X}^*$, $\mathbf{X}^H$, respectively.
$\mathbf{X}^{\dagger} = (\mathbf{X}^H \mathbf{X})^{-1}\mathbf{X}^H$ denotes the pseudo inverse of matrix with linearly independent columns.
The Euclidean norm of $\mathbf{x}$ is given by $\norm{\mathbf{x}}$, and  the Frobenius norm of $\mathbf{X}$ is $\normF{\mathbf{X}}$.
The vectorization of a matrix $\mathbf{X}$ is denoted as $\vect(\mathbf{X})$, while $\invect(\mathbf{x})$ indicates the inverse vectorization process.
A circularly symmetric complex Gaussian random variable with mean $\mu$ and variance $\sigma^2$ is represented as $\mathcal{N}(\mu, \sigma^2)$. The Hadamard product and the Kronecker product of $\mathbf{X}$ and $\mathbf{Y}$ are denoted by $\mathbf{X} \odot \mathbf{Y}$ and $\mathbf{X} \otimes \mathbf{Y}$, respectively. 
The $N \times N$ identity matrix and the $N \times M$ zero matrix are denoted by $\mathbf{I}_{N}$ and $\mathbf{0}_{N \times M}$, respectively. 
The $i$th column and the $(i,j)$th entry of $\mathbf{X}$ are represented as $[\mathbf{X}]_{:,i}$ and $[\mathbf{X}]_{i,j}$,  respectively. 
$\mathrm{blkdiag}(\mathbf{X}_{1}, \dots, \mathbf{X}_{N})$ returns a block diagonal matrix with $\mathbf{X}_{1}, \dots, \mathbf{X}_{N}$ on the main diagonal.

\section{System and Channel Models}\label{System and Channel Models}
In this section, we elaborate on the signal model of the near-field XL-MIMO system and the LoS channel model based on the non-uniform spherical wave model (NUSWM), the uniform spherical wave model (USWM), and the parabolic wavefront model, respectively.
Next, we propose the subarray-wise outer product model (SOPM) for the LoS channel, and review the existing NLoS channel model.

\subsection{System Model}\label{System Model}
We consider a single-carrier partially-connected hybrid MIMO architecture with $\Nr$ receive and $\Nt$ transmit antennas as shown in Fig. \ref{fig:par}.
\begin{figure}[t]
	\centering
	\includegraphics[width = 2.5in]{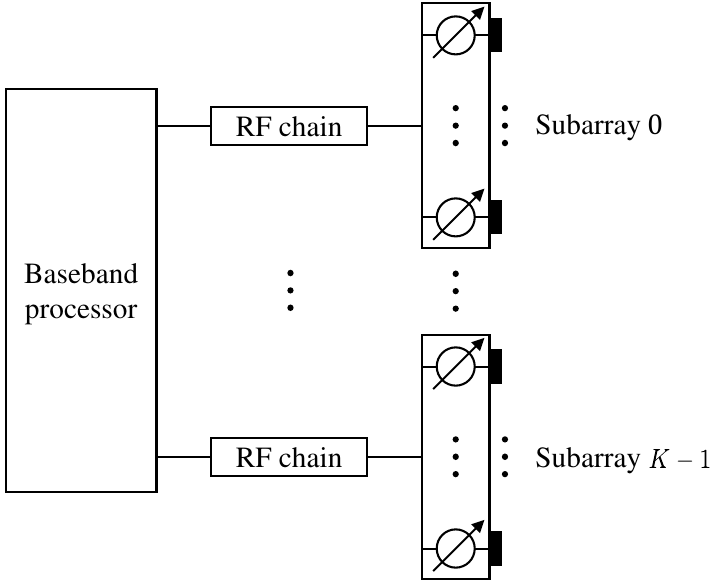}
	\caption{An illustration of partially-connected hybrid architecture.}
	\label{fig:par}
	\vspace{-.2in}
\end{figure}
The receiver and the transmitter are equipped with uniform linear array (ULA) with $ \Kr $ and $ \Kt $ RF chains, respectively, where each receive and transmit RF chain is connected to a disjoint subarray of $ \Nrs = \frac{\Nr}{\Kr}$ receive and $ \Nts = \frac{\Nt}{\Kt} $ transmit antennas, respectively. 
The antenna spacing is $ d = \lambda/2 $, where $ \lambda $ represents the carrier wavelength. 
We represent the channel between  receiver and  transmitter as $ \mathbf{H} \in \mathbb{C}^{\Nr \times \Nt} $. 
Using $ \Mr $ receive and $ \Mt $ transmit beams, the received signal $ \mathbf{Y} \in \mathbb{C}^{ \Mr \times \Mt} $ is written as
\begin{align}
	\mathbf{Y} = \W^{H} \mathbf{H} \F + \W^{H} \mathbf{N},
\end{align}
where $ \W \in \mathbb{C}^{ \Nr \times \Mr } $, $ \F \in \mathbb{C}^{ \Nt \times \Mt } $, and $ \mathbf{N} \in \mathbb{C}^{ \Nr \times \Mt} $ represent the analog combining and precoding matrices, and additive white Gaussian noise whose elements are independent and follow $ \mathcal{CN}(0, \sigma_{w}^2 ) $, respectively.
Representing the number of receive and transmit beams for each subarray as $\Mrs = \frac{\Mr}{\Kr} $ and $\Mts = \frac{\Mt}{\Kt} $, the partially-connected analog combining and precoding matrices are represented as
\begin{align}
	\W & = \blkdiag(\widetilde{\W}_{0}, \dots, \widetilde{\W}_{\Kr-1}) \in \mathbb{C}^{\Nr \times \Mr}, \\
	\F & = \blkdiag(\widetilde{\F}_{0}, \dots, \widetilde{\F}_{\Kt-1}) \in \mathbb{C}^{\Nt \times \Mt},
\end{align}
where $\widetilde{\W}_{i} \in \mathbb{C}^{ \Nrs \times \Mrs } $ and  $\widetilde{\F}_{j} \in \mathbb{C}^{ \Nts \times \Mts } $ denote the combining and precoding matrices of the $i$th transmit and the $j$th receive subarrays, for $ i = 1, \dots, \Kr $, $ j = 1, \dots, \Kt $.
Since the analog combiner and precoder are implemented via phase shifters, their matrices satisfy the constant modulus constraint $ | [\widetilde{\W}_{i}]_{m, n} | = \frac{1}{\Nrs} $ and $ | [\widetilde{\F}_{j}]_{m,n} | = \frac{1}{\Nts} $ for $m = 1, \dots, \Nr$ and $n = 1, \dots, \Nt$.

\subsection{LoS Channel Model} \label{LoS Channel Model}
\subsubsection{Spherical Wave Model}
To accurately reflect the spherical wavefront with non-uniform power received by each antenna element, the NUSWM-based LoS channel between the $ m $th receive and the $ n $th transmit antenna is expressed as 
\begin{align}\label{HNUSWM}
	[\textbf{H}^{\textrm{NUSWM}}\los]_{m,n} & = \frac{\check{g}}{r_{m, n}} e^{-j \frac{2\pi}{\lambda} r_{m, n}},
\end{align}
for $m = 1, \dots, \Nr$, $ n = 1, \dots, \Nt$,
where $ \check{g} $ denotes the small scale channel gain and $ r_{m, n} $ is the distance between the $ m $th receive and the $ n $th transmit antenna. This distance is represented as 
\begin{align}\label{rNUSWM}
	r_{m, n} & \!=\! 
	\Bigl[ \! ( R + \deltarm \!\sin{\thetar}\! \cos{\phir} - \deltatn \sin{\thetat} )^2 \nonumber \\[-5pt]
	& \quad  + ( \deltarm \!\sin{\thetar} \!\sin{\phir} )^2 + ( \deltarm \!\cos{\thetar} - \deltatn \!\cos{\thetat} )^2 \! \Bigr]^{\!1/2}\!,
\end{align}
where $R$ is the distance between the reference receive and transmit antennas,
$ \thetar $  and $ \thetat $ are the elevation angles of the receive and transmit antenna arrays, respectively, and $\phir$ is the azimuth angle of the receive antenna array,  as illustrated in Fig.\ref{fig:ULAs}.
\begin{figure}[t!]
	\centering
	\includegraphics[width = 3.2in]{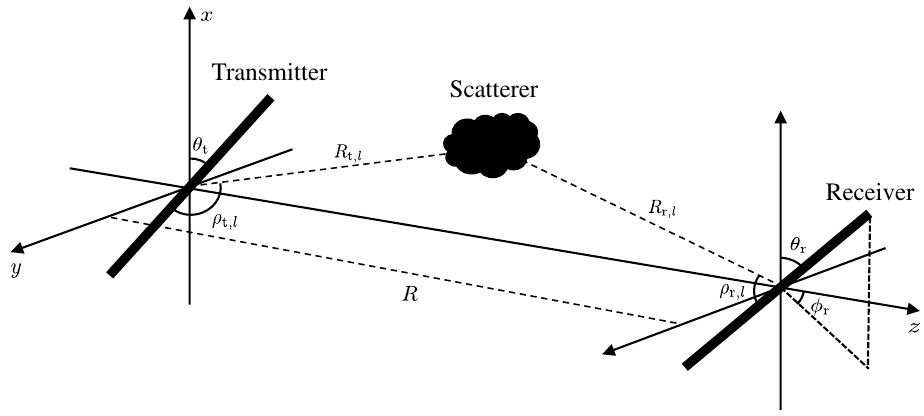}
	\caption{Illustration of transmit and receive ULAs and its coordinate system.}
	\label{fig:ULAs}
	\vspace{-.2in}
\end{figure}
Here, $ \deltarm = (m - \frac{\Nr+1}{2})d $ and $ \deltatn = (n - \frac{\Nt+1}{2})d $.
Note that for arbitrary placement of the transmit and receive arrays, a coordinate system can be constructed in which the azimuth angle of the transmit array is set to zero. Consequently, their relative positions are completely characterized by the three aforementioned angles.

To evaluate the effectiveness of the USWM, \cite{lu2021communicating} proposed a distance criterion termed uniform power distance (UPD) to identify communication distances at which the received signal power does not vary significantly across antenna elements. 
The UPD can be evaluated numerically and is typically small for high-frequency XL-MIMO systems. In this paper, we consider a XL-MIMO system where both the transmitter and receiver are equipped with ULAs of $ \Nr = \Nt = 128 $ at $f_c = 60$ with half-wavelength spacing, to assess the validity of the channel models.
For instance, in the considered XL-MIMO system, the upper bound of the UPD is approximately 12 meters when the threshold for the power ratio is set to $\Gamma_{\textrm{th}} = 0.9$.
Consequently, the NUSWM becomes less relevant in such scenarios.
Hence, for simplicity and practicality in channel modeling, we consider the USWM defined as
\begin{align}\label{HUSWM}
	[\textbf{H}^{\textrm{USWM}}\los]_{m,n}
	& = \frac{\check{g}}{R} e^{-j \frac{2\pi}{\lambda} r_{m, n}},
\end{align}
for $m = 1, \dots, \Nr$, $ n = 1, \dots, \Nt$,
where the received power across antenna element is assumed to be identical.

\subsubsection{Parabolic Wavefront Model}
In the parabolic wavefront model, \( r_{m, n} \) in (\ref{rNUSWM}) is approximated to \( r_{m, n}^{\textrm{parabolic}} \) based on Fresnel approximation as
\begin{align}\label{rparabolic}
r_{m, n}^{\textrm{parabolic}}
& \!=\! R + \deltarm \! \sin{\thetar} \!\cos{\phir} - \deltatn \!\sin{\thetat} \nonumber \\
& \quad + \!\frac{( \deltarm \!\sin{\thetar} \!\sin{\phir} )^2 + ( \deltarm \!\cos{\thetar} - \deltatn \!\cos{\thetat} )^2}{2R}\!,
\end{align}
using a Taylor expansion \( \sqrt{1+x} \approx 1 + \frac{x}{2} \), and neglecting terms much smaller than \( R \) in the denominator. For notational simplicity, we rewrite (\ref{rparabolic}) in terms of transformed parameters ${ \{ \varphir, \alphar, \varphit, \alphat, \eta \} }$ as
\begin{align}\label{rparabolic2}
r_{m, n}^{\textrm{parabolic}}
& \!=\! R + \deltarm^2 \alphar \!-\! \deltatn^2 \alphat \!+\! \deltarm \varphir \!-\! \deltatn \varphit \!-\! \eta \deltarm \deltatn,
\end{align}
where $ \varphir = \sin{\thetar} \cos{\phir} $,  $\varphit = \sin{\thetat} $, 
$ \alphar = \frac{ 1 - \varphir^2 }{2R} $, 
$ \alphat = - \frac{ 1 - \varphit^2 }{2R} $, and 
$ \eta = \frac{\cos{\thetar} \cos{\thetat} }{R} $. 
The validity of the parabolic wavefront model can be verified using the results from \cite{do2023parabolic}. 
To define the communication distance within which the modeling error of the parabolic wavefront model becomes significant, we define the distance at which the largest phase discrepancy between the USWM and the parabolic wavefront model exceeds $ \pi/8 $ as parabolic distance.
This approach is analogous to the definition of the Fraunhofer distance (FD) \cite{selvan2017fraunhofer}.
Fortunately, this distance is typically small, for example, this distance is approximately 4.5 meters in the considered XL-MIMO system. 
Thus, the parabolic wavefront model can be effectively employed in the considered scenario.

To represent the parabolic wavefront model-based LoS channel in matrix form, we define vectors of lengths $\Nr$ and $\Nt$ as follows:
\begin{align}
	[\ar (\varphi, \alpha)]_{m} & = e^{-j \frac{2\pi}{\lambda} ( \deltarm \varphi + \deltarm^2 \alpha )}, \label{steering vector r} \\[-3pt]
	[\at (\varphi, \alpha)]_{n} & = e^{-j \frac{2\pi}{\lambda} ( \deltatn \varphi + \deltatn^2 \alpha )}, \label{steering vector t}
\end{align}
for $m = 1, \dots, \Nr$, $ n = 1, \dots, \Nt$.
Then, the LoS channel matrix based on the parabolic wavefront model can be expressed as
\begin{align}\label{Hparabolic}
	\mathbf{H}\los^{\textrm{parabolic}} = 
	g
	\ar ( \varphir, \alphar ) 
	\at ( \varphit, \alphat )^H
	\odot \boldsymbol{\Lambda}(\eta),
\end{align}
where $g = \frac{\check{g}}{R} e^{-j \frac{2\pi}{\lambda} R } $ and $ \boldsymbol{\Lambda}(\eta) $ represents a $ \Nr \times \Nt $ Vandermonde matrix whose $ (m, n) $th entry is defined as
\begin{align}\label{HNLoS}
	& [\boldsymbol{\Lambda}(\eta)]_{m,n} =  e^{ j \frac{2\pi}{\lambda} \eta \deltarm \deltatn },
\end{align}
for $m = 1, \dots, \Nr$, $ n = 1, \dots, \Nt$.
To evaluate the system performance with respect to the SNR, we represent the squared magnitude of the LoS path gain as $ \lvert g \rvert^2 = \kappa/(1 + \kappa) $, where $\kappa$ is defined as the factor representing the power ratio of the LoS and the NLoS paths.

The phase terms in the parabolic wavefront model, expressed up to the second order using the Fresnel approximation, enable an analytical representation of channel eigenvalues and capacity  \cite{bohagen2005construction, do2020reconfigurable}. However, from the perspective of channel estimation, the presence of the complex coupled phase term $- \eta \deltarm \deltatn$ in (\ref{rparabolic}), or equivalently $\boldsymbol{\Lambda}(\eta)$ in (\ref{Hparabolic}), poses a significant challenge. It complicates the application of efficient and robust channel estimation techniques that rely on the outer product channel matrix model.

\subsubsection{Subarray-Wise Outer Product Model (SOPM)}
To leverage the simplicity of the outer product-based channel model, which enables straightforward dictionary design \cite{lee2016channel} and tensor decomposition-based estimation \cite{zhou2017low}, the parabolic wavefront model-based LoS channel can be approximated as an outer product of the receive and transmit array steering vectors by neglecting $\boldsymbol{\Lambda}(\eta)$ in (\ref{Hparabolic}). 
However, in XL-MIMO systems, the MIMO advanced Rayleigh distance (MIMO-ARD) \cite{lu2023near}, which quantifies the distance boundary where the outer product approximation remains valid, is typically large. This boundary is defined as $ \frac{4 \Ar \At}{\lambda} $, where $\Ar = (\Nr-1)d$ and $\At = (\Nt-1)d$ represent the aperture size of the receive and transmit arrays, respectively.  For example, in the considered XL-MIMO system, the MIMO-ARD is approximately 81 meters. This implies that it typically covers significant portion of the cell coverage.
Therefore, simply neglecting the coupled phase term severely deteriorates the performance in most scenarios in XL-MIMO systems.

To facilitate simpler signal processing based on an accurate channel model, we instead propose the subarray-wise outer product model (SOPM), where the coupled phase term can be accurately approximated as a sum of terms dependent either on transmit or receive antenna index due to the reduced aperture of the subarrays.
In this model, we approximate the coupled distance term in (\ref{rparabolic2}), by utilizing the first-order bivariate Taylor expansion $xy \approx \widebar{x} \widebar{y} + \widebar{y} (x - \widebar{x}) + \widebar{x}(y - \widebar{y}) $, as 
\begin{align}\label{coupled term approximation}
	\eta \deltarm \deltatn & \approx
	\eta \left( \nuri \nutj + \nutj ( \deltarm - \nuri ) + \nuri ( \deltatn - \nutj ) \right) \nonumber \\
	& = \eta ( \nutj \deltarm + \nuri \deltatn - \nuri \nutj ),
\end{align}
where $ i = \left\lceil \frac{m}{\Nrs} \right\rceil $ and $ j = \left\lceil \frac{n}{\Nts} \right\rceil $ are the subarray indices of the $m$th transmit and the $n$th receive antenna, respectively, and $\nuri =  \left( (2i - 1) \Nrs - \Nr \right) d/2$ and $\nutj =  \left( (2j - 1) \Nts - \Nt \right) d/2$ are the centroid of $i$th receive and $j$th transmit subarray, respectively.
Consequently, substituting (\ref{coupled term approximation}) into (\ref{rparabolic2}), the distance between the $ m $th receive and the $ n $th transmit antenna based on the SOPM can be written as
\begin{align}\label{rSOPM}
	r_{m, n}^{\textrm{SOPM}}
	& = R \!+\! \deltarm^2 \alphar \!-\! \deltatn^2 \alphat \nonumber \\
	& \quad +\! \deltarm ( \varphir \!-\! \eta \nutj ) \!-\! \deltatn ( \varphit \!+\! \eta \nuri ) \!+\! \eta \nuri \nutj. \hspace{-2pt}
\end{align}
Since $r_{m, n}^{\textrm{SOPM}}$ does not contain any coupled terms depending both on $m$ and $n$ for antenna indices confined within a subarray, i.e., $ (i-1) \Nrs < m \leq i \Nrs $ and $ (j-1) \Nts < n \leq j \Nts $, the LoS channel matrix based on the SOPM is expressed as
\begin{align}\label{HSOPM}
	\mathbf{H}\los^{\textrm{SOPM}}
	& =
	\begin{bmatrix}
		\widetilde{\mathbf{H}}^{\textrm{SOPM}}_{\textrm{LoS}, 1, 1}	& \cdots &	\widetilde{\mathbf{H}}^{\textrm{SOPM}}_{\textrm{LoS}, 1, \Kt} \\
		\vdots & \ddots & \vdots \\
		\widetilde{\mathbf{H}}^{\textrm{SOPM}}_{\textrm{LoS}, \Kr, 1}	& \cdots &	\widetilde{\mathbf{H}}^{\textrm{SOPM}}_{\textrm{LoS}, \Kr, \Kt}
	\end{bmatrix},
\end{align}
where the $(i, j)$th submatrix of size $\Nrs \times \Nts$ is defined as
\begin{align}
	\widetilde{\mathbf{H}}^{\textrm{SOPM}}_{\textrm{LoS}, i, j}
	& = \tilde{g}_{i, j} 
	\Sri \ar ( \xirj, \alphar ) 
	(\Stj \at ( \xiti, \alphat ) )^H,
\end{align}
where  $\tilde{g}_{i, j} = g e^{-j \frac{2\pi}{\lambda} \eta \nuri \nutj }$, $\Sri = [\mathbf{0}_{\Nrs \times (i-1) \Nrs}, \I_{\Nrs}, \mathbf{0}_{\Nrs \times (\Kr-i)\Nrs}] \in \mathbb{C}^{\Nrs \times \Nr} $, $\Stj = [\mathbf{0}_{\Nts \times (j-1) \Nts}, \I_{\Nts}, \mathbf{0}_{\Nts \times (\Kt-j)\Nts}]  \in \mathbb{C}^{\Nts \times \Nt} $, $ \xirj = \varphir - \eta \nutj $, and $ \xiti = \varphit + \eta \nuri $.

Now, to determine the region where the SOPM remains valid, we first derive the worst case phase error of the SOPM in the following lemma.
\begin{lemma} \rm
	The maximum phase discrepancy between the USWM in (\ref{HUSWM}) and the SOPM at a communication distance $R$ is given by
	\begin{align}
		\Delta \Phi^{\textrm{SOPM}} & \triangleq \max_{ \thetar, \thetat, \phir } \max_{m, n} \frac{2 \pi}{\lambda} \lvert r_{m, n} - r_{m, n}^{\textrm{SOPM}} \rvert \nonumber \\
		& = \frac{\pi \Ars \Ats}{2R \lambda} + \mathcal{O}(R p^4),
	\end{align}
	where $p = \frac{\Ar + \At}{2R}$, and $ \Ars = (\Nrs - 1)d $ and $ \Ats = (\Nts - 1)d $ represent the aperture of the receive and the transmit subarray, respectively.
\end{lemma}
\begin{proof}
	See Appendix \ref{Proof of Lemma 1}.
\end{proof}

Grounded in the worst case phase error analysis, we introduce the subarray-wise outer product distance (SOPD).
This distance is defined as the threshold at which the largest phase discrepancy between the LoS channel based on the USWM and the SOPM exceeds $ \pi/8 $.
Neglecting the higher order phase error terms, the SOPD is given by
\begin{align}
	R_{\textrm{SOPD}} = \frac{4 \Ars \Ats}{ \lambda}.
\end{align}
Since the aperture of the subarrays can be set small, the SPOD can be minimized by using subarrays with small apertures.

\subsection{NLoS Channel Model}
Unlike the LoS channel, the variation in received power levels across each antenna due to the scatterer, i.e., the power ratio, is used to evaluate the effectiveness of the USWM in the NLoS channel.  For instance, in the considered XL-MIMO system, the upper bound of the UPD for the NLoS channel is approximately 6 meters when the power ratio threshold is set to $\Gamma_{\textrm{th}, \mathrm{NLoS}} = 0.9$. Furthermore, since the parabolic distance of the NLoS channel to/from the scatterer is approximately 1.6 meters—a relatively small value—the channel matrix for each NLoS path can be modeled as the outer product of the receive and transmit near-field steering vectors in (\ref{steering vector r}) and (\ref{steering vector t}) \cite{lu2023near}. Thus, the NLoS channel matrix, based on the Fresnel approximation, $ \mathbf{H}\nlos \in \mathbb{C}^{\Nr \times \Nt}$, can be expressed as
\begin{align}\label{key}
	\mathbf{H}\nlos = \sqrt{\frac{1}{L}} \sum_{l = 1}^{ L } g_l  \ar ( \varrhorl, \beta_{\mathrm{r}, l} ) \at ( \varrhotl, \beta_{\mathrm{t}, l} )^H, 
\end{align}
where $ L $ denotes the number of the NLoS paths, $ g_l $ denotes the complex channel gain of the $ l $th NLoS path.
$ \varrhorl = \cos \rhorl $ and $ \varrhotl = - \cos \rhotl $ represent the linear phase terms, and 
$ \beta_{\mathrm{r}, l} = \frac{1 - \varrhorl^2 }{2R_{\mathrm{r}, l}} $ and 
$ \beta_{\mathrm{t}, l} = - \frac{1 - \varrhotl^2 }{2R_{\mathrm{t}, l}} $ are
the quadratic phase terms of the receive and the transmit array steering vectors where $\rhorl$, $\rhotl$, $R_{\mathrm{r}, l}$, and $R_{\mathrm{t}, l}$ represent the angle of arrival, the angle of departure, the distance between the receiver and the scatterer, the distance between the transmitter and the scatterer of the $l$th path, respectively, as shown in Fig. \ref{fig:ULAs}.
The gain of the NLoS paths satisfy $g_{l} \sim \mathcal{CN}(0, \sigma_{l}^2)$, for $l = 1 ,\dots L$, where $ \sigma_{l}^2 = 1/(1+\kappa)$ and $\kappa$ is the factor representing the power ratio of the LoS and the NLoS paths defined in Section \ref{LoS Channel Model}.

\section{Two-stage on-grid Channel Estimation}\label{Proposed Two-stage Channel Estimation}
In this section, we propose a low-complexity two-stage on-grid channel estimation algorithm. 
In the first stage, the LoS channel component is estimated using the alternating subarray-wise array gain maximization (ASAGM) based on the SOPM, treating the NLoS paths as interference since the channel gain of the LoS path is larger than that of the NLoS path. 
This approach allows to reformulate the high-complexity on-grid parameter search problem into several low-dimensional array gain maximization subproblems.  In the second stage, after subtracting the estimated LoS component, the NLoS channel component is estimated using the sensing matrix refinement-based orthogonal matching pursuit (SMR-OMP). 
This technique mitigates the computational complexity associated with support detection based on a high-dimensional joint dictionary by employing the simultaneous orthogonal matching pursuit (SOMP) algorithm, which facilitates independent support detection for the transmit and receive side dictionaries.

\subsection{LoS Channel Estimation}
Since the channel gain of the LoS path is larger than that of the NLoS path, we first estimate the LoS channel. Treating the NLoS channel component as interference, the estimation problem of the LoS channel parameters can be reduced to an on-grid parameter search problem (e.g. maximum likelihood \cite{lu2023near}).
However, the complexity of solving this type of problem is prohibitively large since it involves parameter estimation of high dimension.
To alleviate this issue, we propose reformulating the high-dimensional parameter estimation problem into iterative receive and transmit array gain maximization subproblems with reduced parameter search dimension.
To this end, we represent the received signal of the $i$th receive subarray from the $t$th transmit subarray $\Y_{i, j} \in \mathbb{C}^{\Mrs \times \Mts}$ as 
\begin{align}\label{SOPM system model}
	\Y_{i, j} 
	& = \widetilde{\W}_{i}^H \widetilde{\mathbf{H}}^{\textrm{SOPM}}_{\textrm{LoS}, i, j} \widetilde{\F}_{j} + \widetilde{\W}_{i}^H \mathbf{N}_{i, j} \nonumber \\
	& = \tilde{g}_{i, j} \widetilde{\W}_{i}^H \Sri \ar ( \xirj, \alphar ) 
	(\Stj \at ( \xiti, \alphat ) )^H \widetilde{\F}_{j} \nonumber \\
	& \quad  + \widetilde{\W}_{i}^H \mathbf{N}_{i, j},
\end{align}
where $ \mathbf{N}_{i, j} = \Sri \mathbf{N} \Stj^H \in \mathbb{C}^{\Nrs \times \Mts} $.

Next, to enable correlation-based estimation in the presence of colored noise, a noise-whitening procedure is first  carried out.
The covariance matrix of the noise of the received signal $ \Y_{i, j} $ is $ \K_{i} = \sigma_{w}^2 \widetilde{\W}_{i}^H \widetilde{\W}_{i} \in \mathbb{C}^{\Mrs \times \Mrs} $.
Using Cholesky factorization, the covariance matrix is decomposed as $ \K_{i} = \sigma_{w}^2 \mathbf{L}_{i} \mathbf{L}_{i}^H $.
Then, the whitened received signal is expressed as
\begin{align}\label{SOPM system model}
	\widebar{\Y}_{i, j} & = \mathbf{L}_{i}^{-1} \Y_{i, j} \nonumber \\ 
	& = \tilde{g}_{i, j}  \mathbf{L}_{i}^{-1} \widetilde{\W}_{i}^H \Sri \ar ( \xirj, \alphar ) 
	(\Stj \at ( \xiti, \alphat ) )^H \widetilde{\F}_{j} \nonumber \\
	& \quad  + \mathbf{L}_{i}^{-1} \widetilde{\W}_{i}^H \mathbf{N}_{i, j} \nonumber \\
	& = \tilde{g}_{i, j}  \tildeari ( \xirj, \alphar ) \tildeatj ( \xiti, \alphat )^H + \widebar{\mathbf{N}}_{i, j},
\end{align}
where $ \tildeari ( \xirj, \alpha ) = \mathbf{L}_{i}^{-1} \widetilde{\W}_i^H \Sri \ar ( \xirj, \alphar ) \in \mathbb{C}^{\Mrs \times 1}$, $ \tildeatj ( \xiti, \alphat ) = \widetilde{\F}_j^H \Stj \at ( \xiti, \alphat ) \in \mathbb{C}^{\Mt \times 1} $, and $\widebar{\mathbf{N}}_{i, j} = \mathbf{L}_{i}^{-1} \widetilde{\W}_{i}^H \mathbf{N}_{i, j} \in \mathbb{C}^{\Mrs \times \Mts} $ is the white Gaussian noise whose covariance matrix is $ \widebar{\K} = \sigma_{w}^2 \mathbf{I}_{\Mrs} $

Then, we define the receive and transmit array gain between the $i$th receive and the $j$th transmit subarrays as 
\begin{align}
	\gammarij( \barxirj, \baralphar ) & = \lvert \barari ( \barxirj, \baralphar )^H \tildeari ( \xirj, \alphar ) \rvert, \\
	\gammatij( \barxiti, \baralphat ) & = \lvert \baratj ( \barxiti, \baralphat )^H \tildeatj ( \xiti, \alphat ) \rvert,
\end{align}
where $ \barari ( \barxirj, \baralphar ) = \frac{ \tildeari ( \barxirj, \baralphar ) }{ \norm{\tildeari ( \barxirj, \baralphar )} } $, and $ \baratj ( \barxiti, \baralphat ) = \frac{ \tildeatj ( \barxiti, \baralphat ) }{ \norm{\tildeatj ( \barxiti, \baralphat )} } $.
Since the array gain is maximized when $ ( \barxirj, \baralphar ) =  ( \xirj, \alphar )$ and $( \barxiti, \baralphat ) = ( \xiti, \alphat )$, we can estimate the parameters by identifying those that maximizes the array gain.
Representing  $ G_{i, j}( \barxirj, \baralphar, \barxiti, \baralphat ) \triangleq \lvert \barari ( \barxirj, \baralphar )^H \widebar{\Y}_{i, j} \baratj ( \barxiti, \baralphat ) \rvert $ by the correlation between the array steering vectors and the received signal of the $i$th receive subarray from the $j$th transmit subarray, the noisy observation of the sum of the array gain can be obtained as
\begin{align} 
	& \sum_{i = 1}^{\Kr} \sum_{j = 1}^{\Kt} 
	 G_{i, j}( \barxirj, \barxiti, \baralphar, \baralphat ) \nonumber \\[-4pt]
	& \quad = \sum_{i = 1}^{\Kr} \sum_{j = 1}^{\Kt} 
	\gammarij( \barxirj, \baralphar ) \gammatij( \barxiti, \baralphat ) + \widetilde{\mathbf{N}}_{i, j},
\end{align}
We then cast the channel estimation problem as maximization of the noisy measurement of the array gain written as
\begin{subequations}
\begin{align} 
	\max_ { \barthetar, \barthetat, \barphir, \barR } \, &
	\sum_{i = 1}^{\Kr} \sum_{j = 1}^{\Kt} 
	G_{i, j}( \barxirj, \barxiti, \baralphar, \baralphat )  \label{SOPM PE P} \\
	\text{s.t. }
	\quad &  \barxirj = \barvarphir - d \bareta \nutj, \label{SOPM PE C1} \\
	\quad &  \barxiti = \barvarphit + d \bareta \nuri. \label{SOPM PE C2}
\end{align}
\end{subequations}
As outlined previously, directly solving this problem is highly complex since it involves searching a set of the parameters of high dimension.
Instead, we adopt an alternating approach in which $ \gammarij( \barxirj, \baralphar ) $ and $ \gammatij( \barxiti, \baralphat ) $ are alternately maximized. 
This approach is more susceptible to noise and multipath interference compared to jointly estimating all parameters; however, it can perform effectively if the estimated parameters gradually enhance array gain in each iteration.
Nonetheless, this method encounters challenges when alternately updating the estimates of the parameters $\{ \varphir, \alphar, \eta \}$ and $\{ \varphit, \alphat, \eta \}$ in (\ref{rSOPM}), as $ \gammarij( \barxirj, \baralphar ) $ and $ \gammatij( \barxiti, \baralphat ) $ are coupled through the parameter $ \bareta $, as shown in (\ref{SOPM PE C1}) and (\ref{SOPM PE C2}).
In addition to this coupling, the computational complexity remains high, given that it involves a three-dimensional parameter estimation problem.
To cope with this, we propose ASAGM method, which follows a two-step approach.

In the first step, the parameters $ \{ \boldsymbol{\xi}_{\mathrm{r}}, \alphar \} $ and $ \{ \boldsymbol{\xi}_{\mathrm{t}}, \alphat \} $ are alternately estimated with the constraints (\ref{SOPM PE C1}) and (\ref{SOPM PE C2}) temporarily removed, where $ \boldsymbol{\xi}_{\mathrm{r}} = \{ \xi_{\mathrm{r}, 1}, \xi_{\mathrm{r}, 2}, \dots, \xi_{\mathrm{r}, \Kt} \} $ and $ \boldsymbol{\xi}_{\mathrm{t}} = \{ \xi_{\mathrm{t}, 1}, \xi_{\mathrm{t}, 2}, \dots, \xi_{\mathrm{t}, \Kr} \} $ serve as auxiliary parameters. 
Specifically, in $\tau$th iteration, the parameters are updated as
\begin{align} 
	\{ \hatXir^{\scriptscriptstyle(\tau)}\!, \hatalphar^{\scriptscriptstyle(\tau)}\! \} & \!=\! \argmax_{ \substack{\barXir \in \boldsymbol{\Q}_{\boldsymbol{\xi}\ur} \\ \baralphar \in \boldsymbol{\Q}_{\alphar}} }
	\sum_{i = 1}^{\Kr} \sum_{j = 1}^{\Kt} 
	G_{i, j}( \barxirj, \baralphar, \hatxiti^{\scriptscriptstyle(\tau\!-\!1)}\!, \hatalphat^{\scriptscriptstyle(\tau\!-\!1)}\! )\!, \hspace{-0.5in} \label{ASAGM P1 r} \\[-3pt]
	\{ \hatXit^{\scriptscriptstyle(\tau)}\!, \hatalphat^{\scriptscriptstyle(\tau)}\! \} & \!=\! \argmax_{ \substack{\barXit \in \boldsymbol{\Q}_{\boldsymbol{\xi}\ut} \\ \baralphat \in \boldsymbol{\Q}_{\alphat}}  }
	\sum_{i = 1}^{\Kr} \sum_{j = 1}^{\Kt} 
	G_{i, j}( \hatxirj^{\scriptscriptstyle(\tau)}, \hatalphar^{\scriptscriptstyle(\tau)}, \barxiti, \baralphat )\!, \label{ASAGM P1 t}
\end{align}
where $\boldsymbol{\Q}_{(\cdot)}$ represents the parameter quantization grid of $(\cdot)$.
The problems (\ref{ASAGM P1 r}) and (\ref{ASAGM P1 t}) are $(\Kr + 1)$ and $(\Kt + 1)$ dimensional parameter estimation problems, respectively.
However, since the elements of $ \boldsymbol{\xi}_{\mathrm{r}} $ and $ \boldsymbol{\xi}_{\mathrm{t}} $ are independent, the problems (\ref{ASAGM P1 r}) and (\ref{ASAGM P1 t}) can be equivalently rewritten as sequential parameter search problems of low dimension, where the elements of $ \boldsymbol{\xi}_{\mathrm{r}} $ and $ \boldsymbol{\xi}_{\mathrm{t}} $ are separately estimated for given $ \alphar $ and $ \alphat $.
Specifically,  representing $ \checkxirj^{\scriptscriptstyle(\tau)} (\baralphar) $ as the best estimate of $ \xirj $ for given $ \baralphar $ in $\tau$th iteration, problem (\ref{ASAGM P1 r}) can be reformulated as a sequential parameter estimation problem as follows.
\begin{align} \label{ASAGM P2 r 1}
	\checkxirj^{\scriptscriptstyle(\tau)} (\baralphar) & = 
	\argmax_{\barxirj \in \boldsymbol{\Q}_{\xirj}}
	\sum_{i = 1}^{\Kr}
	G_{i, j}( \barxirj, \baralphar, \hatxiti^{\scriptscriptstyle(\tau-1)}, \hatalphat^{\scriptscriptstyle(\tau-1)} ),
\end{align}
for $j = 1 , \dots, \Kt$, $\forall \baralphar \in \boldsymbol{\Q}_{\alphar}$ and
\begin{align} \label{ASAGM P2 r 2}
\hatalphar^{\scriptscriptstyle(\tau)} & \!=\! 
\argmax_{\baralphar \in \boldsymbol{\Q}_{\alphar}}
\sum_{i = 1}^{\Kr}
\sum_{j = 1}^{\Kt} 
G_{i, j}( \checkxirj^{\scriptscriptstyle(\tau)} (\baralphar), \baralphar, \hatxiti^{\scriptscriptstyle(\tau-1)}\!, \hatalphat^{\scriptscriptstyle(\tau-1)}\! ), \\[-3pt]
\hatxirj^{\scriptscriptstyle(\tau)} & = \checkxirj^{\scriptscriptstyle(\tau)} (\hatalphar^{\scriptscriptstyle(\tau)}). \label{ASAGM P2 r estimate}
\end{align}
Similarly, representing $ \checkxiti^{\scriptscriptstyle(\tau)} (\baralphat) $ as the best estimate of $ \xiti $ for given $ \baralphat $ in $\tau$th iteration, problem (\ref{ASAGM P1 t}) can be recast as 
\begin{align} \label{ASAGM P2 t 1}
	\checkxiti^{\scriptscriptstyle(\tau)} (\baralphat) & =
	\argmax_{\barxiti \in \boldsymbol{\Q}_{\xiti}}
	\sum_{j = 1}^{\Kt} 
	G_{i, j}( \hatxirj^{\scriptscriptstyle(\tau)}, \hatalphar^{\scriptscriptstyle(\tau)}, \barxiti, \baralphat ),
\end{align}
for $i = 1 , \dots, \Kr$, $\forall \baralphat \in \boldsymbol{\Q}_{\alphat}$ and 
\begin{align} \label{ASAGM P2 t 2}
	\hatalphat^{\scriptscriptstyle(\tau)}&  = 
	\argmax_{\baralphat \in \boldsymbol{\Q}_{\alphat}}
	\sum_{i = 1}^{\Kr}
	\sum_{j = 1}^{\Kt} 
	G_{i, j}( \hatxirj^{\scriptscriptstyle(\tau)}, \hatalphar^{\scriptscriptstyle(\tau)}, \checkxiti^{\scriptscriptstyle(\tau)} (\baralphat), \baralphat ), \\[-3pt]
	\hatxiti^{\scriptscriptstyle(\tau)} & = \checkxiti^{\scriptscriptstyle(\tau)} (\hatalphat^{\scriptscriptstyle(\tau)}). \label{ASAGM P2 t estimate}
\end{align}
The first step is computationally efficient, as (\ref{ASAGM P2 r 1}) and (\ref{ASAGM P2 t 1}) consist of $\Kr$ and $\Kt$ parallel one-dimensional parameter estimation problems with respect to each $\xirj$ and $\xiti$ for all $\baralphar \in \boldsymbol{\Q}_{\alphar}$ and $\baralphat \in \boldsymbol{\Q}_{\alphat}$, followed by one-dimensional parameter estimation problem for $\alphar$ and $\alphat$.
A schematic diagram of the parameter search space is illustrated in Fig. \ref{fig:codebook}.
\begin{figure}[t!]
	\centering
	\includegraphics[width = 3.3in]{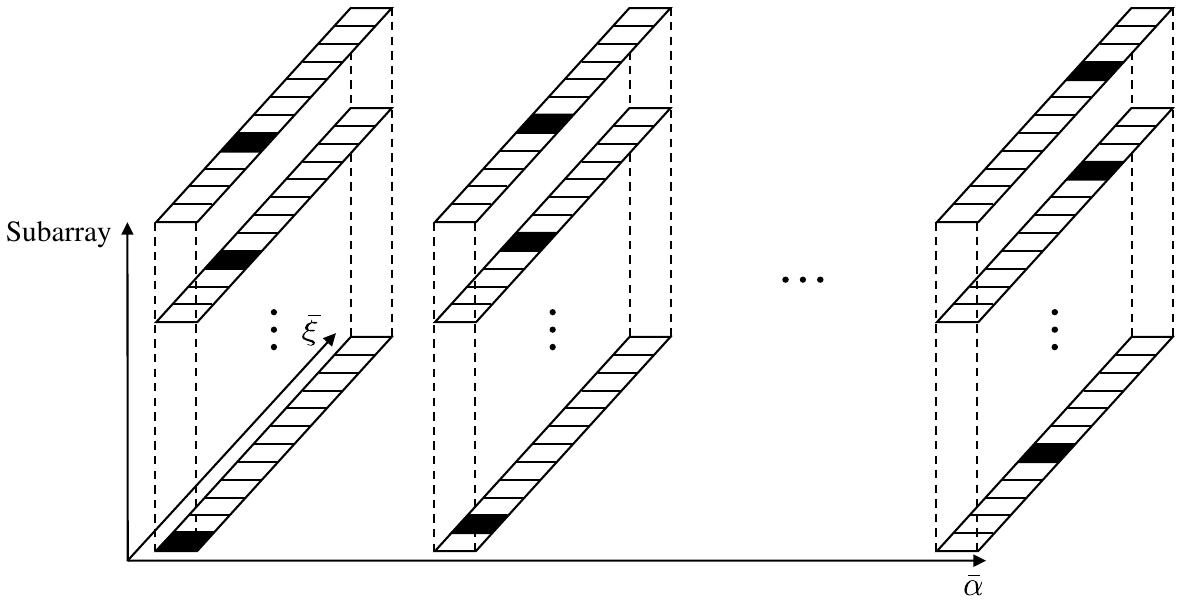}
	\caption{A schematic diagram of the parameter search space.}
	\label{fig:codebook}
	\vspace{-.2in}
\end{figure}
Defining $T_{\textrm{iter}}$ as the number of iterations, the set of the estimated parameters is represented as  $\{ \hat{\boldsymbol{\xi}}_{\mathrm{r}}, \hatalphar, \hat{\boldsymbol{\xi}}_{\mathrm{t}},\hatalphat \} = \{ \hatXir^{\scriptscriptstyle{(T_{\textrm{iter}})}}\!, \hatalphar^{\scriptscriptstyle(T_{\textrm{iter}})}\!, \hatXit^{\scriptscriptstyle{(T_{\textrm{iter}})}}\!, \hatalphat^{\scriptscriptstyle(T_{\textrm{iter}})} \} $.
It is noteworthy that the first step of the ASAGM method converges since the objective value is non-decreasing during each array gain maximization steps (\ref{ASAGM P1 r}) and (\ref{ASAGM P1 t}).

Although the estimates of the auxiliary parameters $\{ \hat{\boldsymbol{\xi}}_{\mathrm{r}}, \hat{\boldsymbol{\xi}}_{\mathrm{t}} \}$ are obtained with low complexity in the first step, the constraints (\ref{SOPM PE C1}) and (\ref{SOPM PE C2}) were relaxed during their estimation.
To incorporate these constraints into the estimation process, we estimate the parameters $\{ \varphir, \varphit, \eta \}$ in the second step by solving a linear regression problem to ensure that the original parameters of interest fit the linear relation of the estimated $\{ \hat{\boldsymbol{\xi}}_{\mathrm{r}}, \hat{\boldsymbol{\xi}}_{\mathrm{t}} \}$.
Specifically, the linear regression problem is formulated by minimizing the sum of the squared error as
\begin{align}\label{SOPM PE LR}
\{ \hatvarphir, \hatvarphit, \hateta \} \nonumber \\[-3pt] 
= \argmin_{ \varphir, \varphit, \eta } & \Biggl( \! \sum_{i=1}^{\Kr} \lvert \hat{\xi}_{\mathrm{t}, i} \!-\! \varphit \!-\! d \eta \nuri \rvert^2 \!+\! \sum_{j=1}^{\Kt} \lvert \hat{\xi}_{\mathrm{r}, j} \!- \!\varphir \!+\! d \eta \nutj \rvert^2 \!\Biggr)\!.
\end{align}
The solution of this problem can be easily obtained, as the squared error is a quadratic function of the parameters.

Finally, based on the estimated parameters $\{ \hatvarphir, \hatalphar, \hatvarphit, \hatalphat, \hateta \} $ and the channel model in (\ref{Hparabolic}), we estimate the channel gain by solving the least squares problem represented as
\begin{align}\label{SOPM PE LS}
	\min_{g} \normF{
		\Y - g \W^H \left(  \ar ( \hatvarphir, \hatalphar ) 
		\at ( \hatvarphit, \hatalphat )^H
		\odot \boldsymbol{\Lambda}(\hateta) \right) \F }^{2}.
\end{align}
Representing the estimated channel gain as $\hat{g}$, the estimated LoS channel matrix is represented as
\begin{align}
	\widehat{\mathbf{H}}\los =
	\hat{g}
	\ar ( \hatvarphir, \hatalphar ) 
	\at ( \hatvarphit, \hatalphat )^H
	\odot \boldsymbol{\Lambda}(\hateta).
\end{align}

The proposed LoS channel estimation is summarized in \textbf{Algorithm \ref{alg:alg1}}.
A detailed analysis of the computational complexity will be provided in Section \ref{Computational Complexity Analysis}.

\begin{algorithm}[t]
	\caption{LoS Channel Estimation Algorithm}\label{alg1} \label{alg:alg1}
	\begin{algorithmic}[1] 
		\renewcommand{\algorithmicrequire}{\textbf{Inputs:}}
		\renewcommand{\algorithmicensure}{\textbf{Initialization:}}
		\Require Whitened received signal $ \widebar{\Y} $, combiner matrix $ \W $, precoder matrix $ \F $, parameter quantization grids $\boldsymbol{\Q}_{\boldsymbol{\xi}\ur}$, $\boldsymbol{\Q}_{\alphar}$, $\boldsymbol{\Q}_{\boldsymbol{\xi}\ut}$, and $\boldsymbol{\Q}_{\alphat}$.
		\For{$\tau \in \{ 1, \dots, T_{\textrm{iter}} \} $}
		\State Update $\checkxirj^{\scriptscriptstyle(\tau)} (\baralphar)$, $j = 1 , \dots, \Kt$, $\forall \baralphar \in \boldsymbol{\Q}_{\alphar}$, via (\ref{ASAGM P2 r 1}).
		\State Update $\hatalphar^{\scriptscriptstyle(\tau)}$ via (\ref{ASAGM P2 r 2}).
		\State Update $\checkxiti^{\scriptscriptstyle(\tau)} (\baralphat)$, $i = 1 , \dots, \Kr$, $\forall \baralphat \in \boldsymbol{\Q}_{\alphat}$, via (\ref{ASAGM P2 t 1}).
		\State Update $\hatalphar^{\scriptscriptstyle(\tau)}$ via (\ref{ASAGM P2 t 2}).
		\EndFor
		\State$\{ \hat{\boldsymbol{\xi}}_{\mathrm{r}},  \hatalphar, \hat{\boldsymbol{\xi}}_{\mathrm{t}},\hatalphat \} = \{ \hatXir^{\scriptscriptstyle{(T_{\textrm{iter}})}},  \hatalphar^{\scriptscriptstyle(T_{\textrm{iter}})}, \hatXit^{\scriptscriptstyle{(T_{\textrm{iter}})}}, \hatalphat^{\scriptscriptstyle(T_{\textrm{iter}})} \} $
		\State Obtain $\{ \hatvarphir, \hatvarphit, \hateta \}$ via (\ref{SOPM PE LR}).
		\State Obtain $\{ \hat{g} \}$ via (\ref{SOPM PE LS}).
		\State $\widehat{\mathbf{H}}\los =
		\hat{g}
		\ar ( \hatvarphir, \hatalphar ) 
		\at ( \hatvarphit, \hatalphat )^H
		\odot \boldsymbol{\Lambda}(\hateta).$
		\Statex \hspace*{-2em} \textbf{Output:} $ \widehat{\mathbf{H}}\los $
	\end{algorithmic} 
\end{algorithm}
%
%
\subsection{NLoS Channel Estimation}
We start by representing the received signal with the estimated LoS component subtracted as
\begin{align}
	\Y\nlos & = \Y - \W^H \widehat{\mathbf{H}}\los \F.
\end{align}
Since the number of scattered paths composing the NLoS channel is usually small in the XL-MIMO systems operating at high frequencies, the NLoS channel can be efficiently estimated via sparse signal recovery method. 
To facilitate this based on the sparse representation of the NLoS channel, we leverage the polar domain receive and transmit side dictionary \cite{cui2022channel} of size $ \Nr \times \QDr $ and $ \Nt \times \QDt $ defined as
\begin{align}
	\Dr & = 
	\bigr[
		\ar( \bar{\varphi}_{\mathrm{r}, 0}, \bar{\alpha}_{\mathrm{r}, 0} ), 
	\dots,
	\ar( \bar{\varphi}_{\mathrm{r}, \QDr - 1} , \bar{\alpha}_{\mathrm{r}, \QDr - 1} )
	\bigr], \\
	\Dt & = 
	\bigr[
		\at( \bar{\varphi}_{\mathrm{t}, 0}, \bar{\alpha}_{\mathrm{t}, 0} ), 
		\dots,
		\at( \bar{\varphi}_{\mathrm{t}, \QDt - 1} , \bar{\alpha}_{\mathrm{t}, \QDt - 1} )
	\bigr],
\end{align}
where $ \QDr $ and $ \QDt $ represent the receive and transmit dictionary quantization levels, respectively.
Then, assuming that the receive and transmit array steering vectors of the scattered paths lie in the quantization grid of the dictionaries, the NLoS channel matrix can be written as 
\begin{align}
	\mathbf{H}\nlos & = \Dr \mathbf{H}\nlosp \Dt^H,
\end{align}
where $ \mathbf{H}\nlosp \in \mathbb{C}^{\QDr \times \QDt}$ denotes the polar domain NLoS channel matrix.
For sparse vector representation, we represent the vectorized NLoS channel as $ \mathbf{h}\nlos = \vect( \mathbf{H}\nlos ) \in \mathbb{C}^{\QDr \QDt \times 1 } $. Then, using the relation $\vect(\A \mathbf{X} \B) = (\B^T \otimes \A) \vect(\mathbf{X})$, the vectorized channel is expressed as 
\begin{align}\label{nlos joint dict}
	\mathbf{h}\nlos  & = ( \Dt^* \otimes \Dr ) \vect(\mathbf{H}\nlosp) \nonumber \\
	& = \boldsymbol{\Psi} \mathbf{h}\nlosp
\end{align}
where $ \boldsymbol{\Psi} = \Dt^* \otimes \Dr \in \mathbb{C}^{\Nr \Nt \times \QDr \QDt } $ and $ \mathbf{h}\nlosp = \vect(\mathbf{H}\nlosp) \in \mathbb{C}^{\QDr \QDt \times 1} $ represents the polar domain joint dictionary and polar domain channel vector.

Next, assuming that the LoS channel component is perfectly subtracted, the received signal for the NLoS channel can be written as
\begin{align}\label{nlos sparse mat}
	\Y\nlos = \W^H \Dr \mathbf{H}\nlosp \Dt^H \F + \W^H \mathbf{N}.
\end{align}
Since the noise is colored with covariance matrix $ \K = \sigma_{w}^2 \W^H \W \in \mathbb{C}^{\Mr \times \Mr} $, we whiten the received signal of the NLoS component based on the Cholesky decomposition $\K = \sigma_{w}^2 \mathbf{L} \mathbf{L}^H$ as
\begin{align}\label{nlos sparse mat whitened}
	\widebar{\Y}\nlos & = \mathbf{L}^{-1} \Y\nlos \nonumber \\
	& = \mathbf{L}^{-1} \W^H \Dr \mathbf{H}\nlosp \Dt^H \F + \widebar{\mathbf{N}},
\end{align}
where $\widebar{\mathbf{N}} = \mathbf{L}^{-1} \W^H \mathbf{N} \in \mathbb{C}^{\Mr \times \Mt}$ is the white Gaussian noise.
For sparse vector representation, we define $ \bar{\mathbf{y}}\nlos = \vect( \widebar{\mathbf{Y}}\nlos ) \in \mathbb{C}^{\Mr \Mt \times 1} $, which can be represented as 
\begin{align}\label{nlos sparse vect}
	\bar{\y}\nlos & = (\F^T \otimes \mathbf{L}^{-1} \W^H) \boldsymbol{\Psi} \mathbf{h}\nlosp + \bar{\mathbf{n}} \nonumber \\
	& = \boldsymbol{\Phi} \boldsymbol{\Psi} \mathbf{h}\nlosp + \bar{\mathbf{n}} \nonumber \\
	& = \boldsymbol{\Upsilon} \mathbf{h}\nlosp + \bar{\mathbf{n}},
\end{align}
where $ \boldsymbol{\Phi} = \F^T \otimes \mathbf{L}^{-1} \W^H \in \mathbb{C}^{ \Mr \Mt\times \Nr \Nt } $, $ \boldsymbol{\Upsilon} = \boldsymbol{\Phi} \boldsymbol{\Psi} \in \mathbb{C}^{ \Mr \Mt\times \QDr \QDt } $ and $ \bar{\mathbf{n}} = \vect(\widebar{\mathbf{N}}) \in \mathbb{C}^{\Mr \Mt \times 1} $ denote measurement matrix, sensing matrix, and noise vector, respectively. 
Since $\mathbf{h}\nlosp$ is a sparse vector, it can be efficiently estimated via sparse signal recovery methods.
For example, orthogonal matching pursuit (OMP) can be employed to estimate the NLoS channel based on (\ref{nlos sparse vect}) with a computational complexity of $\mathcal{O}( \widehat{L} \QDr \QDt \Mr \Mt)$, where $\widehat{L}$ is the estimated number of the NLoS paths. 
This complexity grows linearly with the dimension of the sensing matrix $\QDr \QDt$.
However, in the near-field XL-MIMO systems, $\QDr$ and $\QDt$ are often very large matrices since geometric parameters are quantized in both angular and distance domain to generate a polar domain dictionary.

To mitigate the prohibitive computational complexity arising from the large quantization level of the joint dictionary $\QDr \QDt$, we propose sensing matrix refinement-based OMP (SMR-OMP) algorithm.
In this algorithm, we obtain a refined sensing matrix with reduced dimension through low-complexity coarse support detection in the first step.
In the second step, we employ the OMP algorithm to accurately estimate the sparse channel vector with low computational complexity proportional to the dimension of the reduced sensing matrix.
For low-complexity coarse support detection in the first step, we exploit the row and column sparse structure of the transformed channel, where only the receive or transmit side is in the polar domain.
Specifically, we utilize simultaneous OMP (SOMP) algorithm \cite{tropp2005simultaneous} illustrated in \textbf{Algorithm \ref{alg:alg2}}.
The SOMP algorithm efficiently detects structured support by comparing the sum of coherence across consecutive elements instead of single element as described in line 2 of \textbf{Algorithm \ref{alg:alg2}}.

\begin{algorithm}[t]
	\caption{SOMP}\label{alg:alg2}
	\begin{algorithmic}[1]
		\renewcommand{\algorithmicrequire}{\textbf{Inputs:}}
		\renewcommand{\algorithmicensure}{\textbf{Initialization:}}
		\Require Received signal $ \Y $, sensing matrix $\boldsymbol{\Upsilon}$, sparsity $ L $.
		\Ensure Support set $ \boldsymbol{\Omega} = \emptyset$, residual matrix $ \mathbf{R} = \Y $.
		\For {$ l \in \{1 ,\dots, L \} $}
		\State Detect new support: $ n^{\star} = \argmax_{n} \norm{ [\boldsymbol{\Upsilon}^H \mathbf{R}]_{n, :} }_{2} $.
		\State Update support set: $ \boldsymbol{\Omega} = \boldsymbol{\Omega} \cup \{ n^{\star} \}$.
		\State Update coefficient: $ \widehat{\C} = ( [ \boldsymbol{\Upsilon} ]_{:, \boldsymbol{\Omega}} )^{\dagger} \Y $.
		\State Update residual: $ \mathbf{R} = \Y - [ \boldsymbol{\Upsilon} ]_{:, \boldsymbol{\Omega}} \widehat{\C} $.
		\EndFor
		\Statex \hspace*{-2em} \textbf{Output:} $\boldsymbol{\Omega}$.
	\end{algorithmic} 
\end{algorithm}

To reformulate the NLoS channel support detection problem as a sparse signal recovery problem with multiple observations based on one-dimensional dictionaries, we recast (\ref{nlos sparse mat whitened}) as 
\begin{align}
	\widebar{\Y}\nlos = \boldsymbol{\Upsilon}_{\mathrm{r}} \C_{\mathrm{r}} + \widebar{\mathbf{N}},
\end{align}
where $ \boldsymbol{\Upsilon}_{\mathrm{r}} = \mathbf{L}^{-1} \W^H \Dr \in \mathbb{C}^{\Mr \times \QDr }$, and $\C_{\mathrm{r}} = \mathbf{H}\nlosp \Dt^H \F\in \mathbb{C}^{\QDr \times \Mt}$ is a transformed sparse channel matrix with $L$-nonzero rows.
Due to the row-sparsity of $\C_{\mathrm{r}}$, we can efficiently detect the support set of $\Dr$ using the SOMP algorithm.
Similarly, we can detect the support set of $\Dt$ based on the formulation 
\begin{align}
	\widebar{\Y}^H\nlos = \boldsymbol{\Upsilon}_{\mathrm{t}} \C_{\mathrm{t}} +  \widebar{\mathbf{N}}^H,
\end{align}
where $\boldsymbol{\Upsilon}_{\mathrm{t}} = \F^H \Dt \in \mathbb{C}^{ \Mt \times \QDt }$ and $\C_{\mathrm{t}} = \mathbf{H}\nlosp^H \Dr^H \W (\mathbf{L}^{-1})^H \in \mathbb{C}^{\QDt \times \Mr}$ is a transformed sparse channel matrix with $L$-nonzero rows.

Denoting $ \hatOmegar $ and $ \hatOmegat $ by the detected support sets of $\Dr$ and $\Dt$, respectively, we represent the refined dictionaries for receive and transmit array steering vectors as 
$\hatDr = [\Dr]_{:, \hatOmegar} \in \mathbb{C}^{\Nr \times \widehat{L}_{\mathrm{r}} }$ and $\hatDt = [\Dt]_{:, \hatOmegat} \in \mathbb{C}^{\Nt \times \widehat{L}_{\mathrm{t}} }$,
where $\widehat{L}_{\mathrm{r}} $ and $\widehat{L}_{\mathrm{t}}$ denote the estimated number of NLoS paths from the scatterer to the receiver and from the transmitter to the scatterer, respectively.
Based on this, we represent the refined joint dictionary similar to (\ref{nlos sparse vect}) as 
\begin{align}
	\widehat{\boldsymbol{\Psi}} = \hatDt^* \otimes \hatDr \in \mathbb{C}^{\Nr \Nt \times \widehat{L}_{\mathrm{r}} \widehat{L}_{\mathrm{t}} }.
\end{align}
Then, assuming that the support set is correctly detected in the coarse estimation step, (\ref{nlos sparse vect}) is rewritten as
\begin{align}\label{key}
	\bar{\y}\nlos
	& = \boldsymbol{\Phi} \widehat{\boldsymbol{\Psi}} \mathbf{h}\nlospp + \bar{\mathbf{n}} \nonumber \\
	& = \widehat{\boldsymbol{\Upsilon}}  \mathbf{h}\nlospp + \bar{\mathbf{n}},
\end{align}
where $ \widehat{\boldsymbol{\Upsilon}} = \boldsymbol{\Phi} \widehat{\boldsymbol{\Psi}} \in \mathbb{C}^{\Mr \Mt \times \widehat{L}_{\mathrm{r}} \widehat{L}_{\mathrm{t}} } $ is the refined sensing matrix and $ \mathbf{h}\nlospp $ is the channel vector in the refined polar domain.
With the refined sensing matrix of low dimension $\widehat{\boldsymbol{\Upsilon}}$, the channel matrix can be recovered via OMP algorithm with a computational complexity of $\mathcal{O}( \max(\widehat{L}_{\mathrm{r}}, \widehat{L}_{\mathrm{t}}) \widehat{L}_{\mathrm{r}} \widehat{L}_{\mathrm{t}} \Mr \Mt)$, mainly dominated by the correlation calculation step.
Since $ \widehat{L}_{\mathrm{r}} \widehat{L}_{\mathrm{t}} \ll \QDr \QDt$, the computational complexity of the proposed SMR-OMP is greatly reduced compared to the OMP algorithm with joint dictionary.
Finally, the estimated NLoS channel matrix can be given by
\begin{align}
	 \widehat{\mathbf{H}}\nlos = \invect( [ \widehat{\boldsymbol{\Psi}} ]_{:, \boldsymbol{\widehat{\Omega}}} \widehat{\mathbf{h}}\nlospp ),
\end{align}
where $ \boldsymbol{\widehat{\Omega}} $ and $\widehat{\mathbf{h}}\nlospp $ represent the detected support set of the refined sensing matrix $\widehat{\boldsymbol{\Psi}}$ and the estimated channel vector in the refined polar domain via OMP algorithm.
The SMR-OMP-based NLoS channel estimation algorithm is summarized in \textbf{Algorithm 3}.
Finally, the estimated channel is obtained as 
\begin{align}
	\widehat{\mathbf{H}} & = \widehat{\mathbf{H}}\los + \widehat{\mathbf{H}}\nlos.
\end{align}

The flowchart of the proposed LoS and NLoS channel estimation algorithm is presented in Fig. \ref{fig:flowchart}.
\begin{figure*}[t!]
\centering
\includegraphics[width = 6.5in]{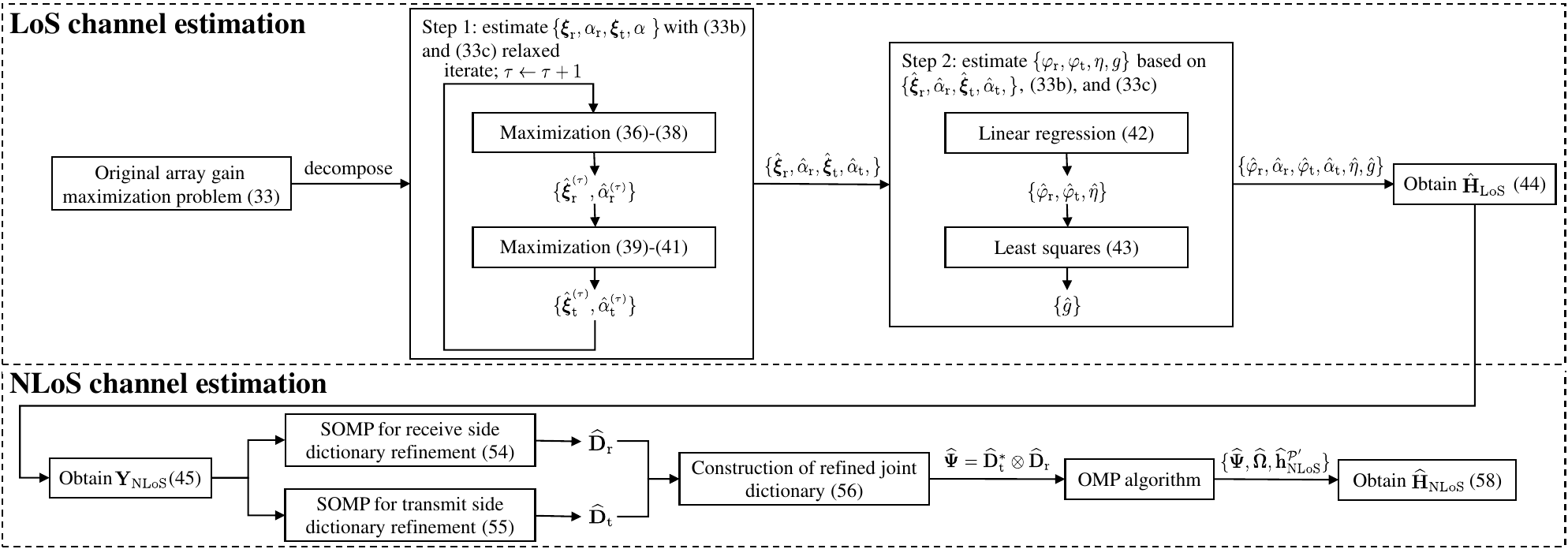}
\caption{Flowchart of the proposed two-stage channel estimation algorithm.}
\label{fig:flowchart}
\vspace{-.2in}
\end{figure*}
A comprehensive analysis of the computational complexity will be presented in Section \ref{Computational Complexity Analysis}.

\begin{algorithm}[t]
\caption{SMR-OMP based NLoS Channel Estimation}\label{alg:alg3}
\begin{algorithmic}[1]
	\renewcommand{\algorithmicrequire}{\textbf{Inputs:}}
	\renewcommand{\algorithmicensure}{\textbf{Initialization:}}
	\Require Received signal for NLoS component $ \widebar{\Y}\nlos $, sensing matrices $\boldsymbol{\Upsilon}_{\mathrm{r}}$ and $\boldsymbol{\Upsilon}_{\mathrm{t}}$, estimated number of the NLoS paths $\widehat{L}_{\mathrm{r}}$ and $\widehat{L}_{\mathrm{t}}$, and measurement matrix $\boldsymbol{\Phi}$.
	\Ensure Support set of the refined dictionary $\widehat{\boldsymbol{\Omega}} = \emptyset$, residual vector $ \mathbf{r} = \bar{\y}\nlos $.
	\State Detect support set of $ \Dr $ and $ \Dt $: 
	\Statex  $\hatOmegar = \textrm{SOMP}(\widebar{\Y}\nlos, \boldsymbol{\Upsilon}_{\mathrm{r}}, \widehat{L}_{\mathrm{r}}) $
	\Statex  $\hatOmegat = \textrm{SOMP}(\widebar{\Y}\nlos^H, \boldsymbol{\Upsilon}_{\mathrm{t}}, \widehat{L}_{\mathrm{t}}) $
	\State Obtain refined joint dictionary: $\widehat{\boldsymbol{\Psi}} = \hatDt^* \otimes  \hatDr$.
	\State Obtain refined sensing matrix: $ \widehat{\boldsymbol{\Upsilon}} = \boldsymbol{\Phi} \widehat{\boldsymbol{\Psi}} $.
	\For {$ l \in \{1 ,\dots, \max(\widehat{L}_{\mathrm{r}}, \widehat{L}_{\mathrm{t}}) \} $}
	\State Detect new support: $ n^{\star} = \argmax_{n} \lvert [\widehat{\boldsymbol{\Upsilon}}^H \mathbf{r}]_{n} \rvert $.
	\State Update support set: $ \widehat{\boldsymbol{\Omega}} = \widehat{\boldsymbol{\Omega}} \cup \{ n^{\star} \}$.
	\State Update coefficient: $ \widehat{\mathbf{h}}\nlospp  = ( [ \widehat{\boldsymbol{\Upsilon}} ]_{:, \widehat{\boldsymbol{\Omega}}} )^{\dagger} \bar{\y}\nlos $.
	\State Update residual: $ \mathbf{r} = \bar{\y}\nlos - [ \widehat{\boldsymbol{\Upsilon}} ]_{:, \widehat{\boldsymbol{\Omega}}} \widehat{\mathbf{h}}\nlospp $.
	\EndFor
	\State $\widehat{\mathbf{H}}\nlos = \invect( [ \widehat{\boldsymbol{\Psi}} ]_{:, \boldsymbol{\widehat{\Omega}}} \widehat{\mathbf{h}}\nlospp ).$
	\Statex \hspace*{-2em} \textbf{Output:} $ \widehat{\boldsymbol{\Psi}} $, $ \widehat{\boldsymbol{\Omega}} $, $ \widehat{\mathbf{h}}\nlospp $.
\end{algorithmic} 
\end{algorithm}

\subsection{Computational Complexity Analysis} \label{Computational Complexity Analysis}
We provide the analysis of the computational complexity of the LoS and NLoS channel estimation of the proposed scheme as well as the benchmark scheme \cite{lu2023near}. Additionally, we evaluate the computational complexity of the joint LoS and NLoS channel estimation approaches outlined in \cite{cui2022channel, shi2024double}.
\subsubsection{LoS channel estimation}
We represent the quantization level of the parameter $(\cdot)$ as $Q_{(\cdot)}$. Additionally, we assume that $ Q_{\xi_{\mathrm{r}}} \triangleq Q_{\xi_{\mathrm{r}, 1}} = \cdots = Q_{\xi_{\mathrm{r}, \Kt}}$ and $ Q_{\xi_{\mathrm{t}}} \triangleq Q_{\xi_{\mathrm{t}, 1}} = \cdots = Q_{\xi_{\mathrm{t}, \Kr}}$.
The computational complexity of the proposed LoS channel estimation algorithm is mainly dominated by steps 1-6 in \textbf{Algorithm \ref{alg:alg1}}, which involve iterative updates of $\{ \hatXir^{\scriptscriptstyle(\tau)}, \hatalphar^{\scriptscriptstyle(\tau)} \}$ in (\ref{ASAGM P2 r 1})-(\ref{ASAGM P2 r estimate}) and $\{ \hatXit^{\scriptscriptstyle(\tau)}, \hatalphat^{\scriptscriptstyle(\tau)} \}$ in (\ref{ASAGM P2 t 1})-(\ref{ASAGM P2 t estimate}). 
The computational complexities of these steps are $ \mathcal{O} ( T_{\textrm{iter}} \Kt \Mr Q_{\xi_{\mathrm{r}}} Q_{\alphar} ) $ and $ \mathcal{O} ( T_{\textrm{iter}} \Kr \Mt Q_{\xi_{\mathrm{t}}} Q_{\alphat} ) $, which are mainly due to (\ref{ASAGM P2 r 1}) and (\ref{ASAGM P2 t 1}), respectively. 
Steps 8 and 9 have complexities of $\mathcal{O}(\Kr + \Kt)$ and $\mathcal{O}(\Mr \Mt \min(\Mr, \Mt) )$, respectively.
Since the parameter quantization levels are usually much larger than the number of the training beams $\Mr$ and $\Mt$, the overall complexity of the proposed LoS estimation algorithm is $\mathcal{O}(  T_{\textrm{iter}} (\Kt \Mr Q_{\xi_{\mathrm{r}}} Q_{\alphar} + \Kr \Mt Q_{\xi_{\mathrm{t}}} Q_{\alphat}) )$.
In contrast, the complexity of the on-grid LoS estimation algorithm based on the geometric parameters \cite{lu2023near} is $\mathcal{O}( \Mr \Mt Q_{\thetar} Q_{\thetat} Q_{\phir} Q_{R})$.
In addition, the complexity of the gradient descent based refinement step is $\mathcal{O}( T_{\textrm{grad}} \Mr \Mt \Nr \Nt )$, where $T_{\textrm{grad}}$ denotes the number of the iteration of the gradient descent step.
\subsubsection{NLoS channel estimation}
The complexity of the proposed NLoS channel estimation algorithm can be obtained by examining the SMR-OMP algorithm in \textbf{Algorithm \ref{alg:alg3}}.
For simplicity, we assume that $\widehat{L} \triangleq \widehat{L}_{\mathrm{r}} = \widehat{L}_{\mathrm{t}} $.
In step 1, the complexity of the SOMP algorithm is dominated by support detection and coefficient update steps, shown in steps 2 and 4 in \textbf{Algorithm \ref{alg:alg2}} which are $ \mathcal{O}( \widehat{L} \Mr \Mt (  \QDr + \QDt ) )$ and $ \mathcal{O}( \widehat{L}^2  (\Mr  + \Mt ) ) $, respectively.
Since the quantization level is often much larger than the number of the paths, i.e., $ \widehat{L} \ll \QDr, \QDt $ the complexity is $ \mathcal{O}( \widehat{L} \Mr \Mt (  \QDr + \QDt ) )$.
The complexities of correlation calculation and coefficient update in steps 6 and 8 are $ \mathcal{O}( \widehat{L}^{3} \Mr \Mt  ) $ and $ \mathcal{O}( \widehat{L}^{2} \Mr \Mt  ) $, respectively.
Therefore, the complexity of the proposed SMR-OMP is $ \mathcal{O}( \widehat{L} \Mr \Mt (  \QDr + \QDt ) )$.
In contrast, the computational complexity of the conventional OMP algorithm with joint dictionary (\ref{nlos joint dict}) is $\mathcal{O}( \widehat{L} \Mr \Mt \QDr \QDt )$ \cite{cui2022channel, lu2023near}, which is dominated by the correlation calculation step.
\subsubsection{Joint LoS and NLoS channel estimation}
The computational complexity of the conventional OMP algorithm for joint LoS and NLoS channel estimation \cite{cui2022channel} is $\mathcal{O}( (\widehat{L}+1) \Mr \Mt \QDr \QDt )$.
Additionally, the complexity of the 3-stage multiple measurement vector unified OMP scheme \cite{shi2024double} for joint LoS and NLoS channel estimation is dominated by the third stage where the $Q_{\eta}$ matrices of size $\Nr \times \Nt$ are transformed to polar domain signal. The computational complexity of this step is $\mathcal{O}( (\widehat{L}+1) Q_{\eta} \QDr \QDt \min(\Nr, \Nt) )$. 

Overall, the computational complexity of the proposed method grows at a much slower rate with respect to the quantization level. Consequently, it can accommodate significantly higher quantization levels, allowing the proposed algorithm to achieve substantial performance improvements. This is particularly important, as estimation accuracy is highly dependent on the parameter quantization level. A summary of the computational complexities for the proposed and benchmark schemes is provided in Table \ref{table:complexity}.

\begin{table}[t]
	\centering
	\caption{Computational Complexity Comparison} \label{table:complexity}
	\begin{tabular}{|ll|l|}
		\hline
		\multicolumn{2}{|l|}{Estimation scheme}                       & Computational complexity  \\ \hline
		\multicolumn{1}{|l|}{\multirow{2}{*}{LoS}}&Proposed&$\mathcal{O}(  T_{\textrm{iter}} (\Kt \Mr Q_{\xi_{\mathrm{r}}} Q_{\alphar} + \Kr \Mt Q_{\xi_{\mathrm{t}}} Q_{\alphat}) )$  \\ \cline{2-3} 
		\multicolumn{1}{|l|}{}                   & \cite{lu2023near}  & $\mathcal{O}( \Mr \Mt (Q_{\thetar} Q_{\thetat} Q_{\phir} Q_{R} + T_{\textrm{grad}} \Nr \Nt ) ))$  \\ \hline
		\multicolumn{1}{|l|}{\multirow{2}{*}{NLoS}} & Proposed  & $\mathcal{O}( \widehat{L} \Mr \Mt (  \QDr + \QDt ) )$  \\ \cline{2-3} 
		\multicolumn{1}{|l|}{}                   & \cite{lu2023near} & $\mathcal{O}( \widehat{L} \Mr \Mt \QDr \QDt )$ \\ \hline
		\multicolumn{1}{|l|}{Joint}              & \cite{cui2022channel} & $\mathcal{O}( (\widehat{L}+1) \Mr \Mt \QDr \QDt )$ \\ \hline
		\multicolumn{1}{|l|}{Joint}              & \cite{shi2024double} & $\mathcal{O}( (\widehat{L}+1) Q_{\eta} \QDr \QDt \min(\Nr, \Nt) )$ \\ \hline
	\end{tabular}
	\vspace{-.2in}
\end{table}

\section{Simulation Results}\label{Simulation Results}
In this section, we evaluate the performance of the proposed two-stage channel estimation algorithm for the near-field partially-connected XL-MIMO with the ASAGM-based LoS channel estimation and the SMR-OMP-based NLoS channel estimation algorithm (``ASAGM + SMR-OMP'').
We compare the proposed method with the existing  maximum likelihood-based LoS parameter estimation and NLoS estimation via OMP with the joint near-field dictionary defined in (\ref{nlos joint dict}) (``PE + OMP'') \cite{lu2023near}, 3-stage multiple measurement vector unified OMP (``3S-MMV-UOMP'') \cite{shi2024double}, OMP with the joint near-field dictionary-based channel estimation algorithm \cite{cui2022channel}, and Genie-aided least square (``Genie-aided LS'') method where the perfect channel parameter knowledge is available.
The performance evaluation is based on the normalized mean square error (NMSE), defined as $ \text{NMSE} = \mathbb{E} \left[ \frac{  \norm{\mathbf{\hat{H}} - \mathbf{H}}_F^2 }{\norm{\mathbf{H}}_F^2} \right] $, where $\mathbb{E}(\cdot)$ represents the expectation operator. 

In \cite{lu2023near}, only the scenario where the azimuth of the receive antenna array $ \phir $ is fixed at $0$ is considered, resulting in a reduced computational complexity of the LoS channel estimation algorithm. 
In contrast, the scenario considered in this paper allows $ \phir $ to be an arbitrary value, making the computational complexity of the LoS channel estimation for the PE + OMP algorithm computationally infeasibly high.
Therefore, for comparison with the PE + OMP algorithm under feasible computation, we employ a partially Genie-aided parameter estimation (``partially Genie-aided PE'') method where the neighborhood of the true parameter values in the parameter quantization grid is known a priori.

We consider a partially-connected hybrid XL-MIMO system where both the transmitter and receiver are equipped with ULA, each having $\Nr = \Nt = 128$ antennas and $\Kr =4$ and $\Kt = 2$ RF chains.
The carrier frequency is $f_c = 60$ GHz, corresponding to a wavelength $\lambda = 0.005$ meters. 
In this scenario, the SOPD is given by $\frac{4 \Ars \Ats}{\lambda} = 9.5$ meters, and the MIMO-ARD is given by $\frac{4 A_{\mathrm{r}} A_{\mathrm{t}}}{\lambda} = 80.6$ meters.
The angle parameters of the ULAs, $\thetat$, $\thetar$, and $\phir$ are uniformly sampled from $[-60^{\circ}, 60^{\circ}]$. 
Additionally, the minimum and maximum allowable distances between the transmitter and receiver are set to 10 meters and 200 meters, respectively.
We define the SNR as $1/\sigma_{w}^2$. 
The number of NLoS paths is $L = 3$, and the  power ratio factor of the LoS and the NLoS paths is $\kappa = 4$, indicating that the LoS channel gain satisfies $|g|^2 = 4/5$,  while the NLoS channel gains follow $g_{l} \sim \mathcal{CN}(0, 1/15)$, for $l = 1, \dots, L$. 
Importantly, while the channel estimation problems are formulated based on the approximated models, the channel used in the simulation is generated using the NUSWM (\ref{HNUSWM}).

The parameter quantization levels of the proposed ASAGM is set as $ Q_{\xi_{\mathrm{r}}} = Q_{\xi_{\mathrm{t}}} = 640 $ and $ Q_{\alphar} = Q_{\alphat} = 7 $. 
For the PE algorithm described in \cite{lu2023near}, the levels are $ Q_{\thetar} = Q_{\thetat} = Q_{\phir} = 196$ and $ Q_{R} = 256$.
In the partially Genie-aided PE algorithm, we limit the search grid to 5 neighboring points around the grid point closest to the true parameter in the parameter quantization grid.
The quantization levels of the transmit and receive polar domain dictionaries are set as $\QDr = \QDt = 1792 $, respectively.
Note that even for large values of $Q_{\xi_{\mathrm{r}}}$ and $Q_{\xi_{\mathrm{t}}}$, the dimensionality of the search grid in the proposed method, which is proportional to $Q_{\xi_{\mathrm{r}}} Q_{\alphar}$ and $Q_{\xi_{\mathrm{t}}} Q_{\alphat}$ are much smaller than that of the benchmark schemes, such as $Q_{\thetar} Q_{\thetat} Q_{\phir} Q_{R}$ in PE + OMP, $Q_{\eta} \QDr \QDt$ in 3S-MMV-UOMP, or $\QDr \QDt$ in OMP, as detailed in Section \ref{Computational Complexity Analysis}. Consequently, the computational complexity of the proposed scheme is considerably lower.

\begin{figure}[t]
	\centering
	\includegraphics[width = 2.8in]{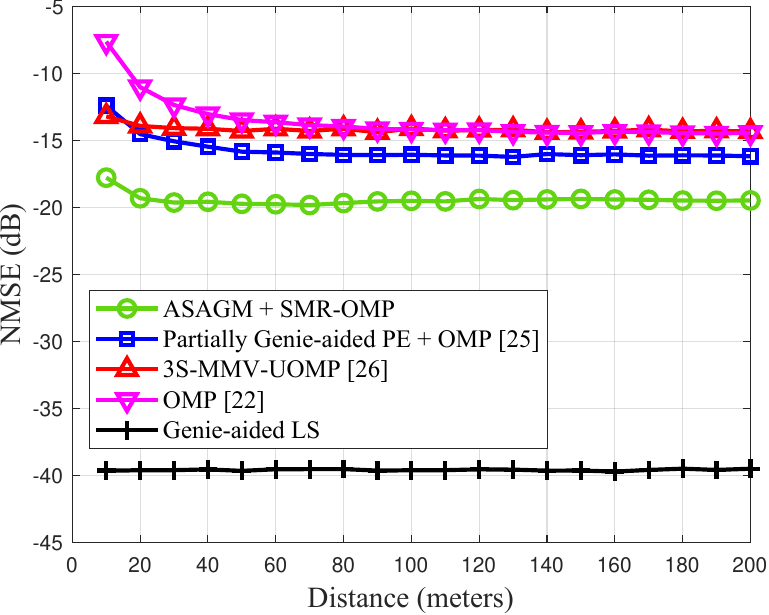}
	\caption{NMSE versus the distance between the transmitter and the receiver.} 
	\label{fig:sim_vs_distance}
	\vspace{-.3in}
\end{figure}
Fig. \ref{fig:sim_vs_distance} illustrates the NMSE performance versus the distance between the transmitter and the receiver where the SNR is $10$ dB and the number of receive and transmit training beams is $\Mr = \Mt = 64$. 
The proposed ASAGM + SMR-OMP method achieves superior NMSE performance compared to other three considered methods.
Notably, the NMSE of the proposed algorithm remains consistently the lowest across all distances evaluated, except for the Genie-aided LS method. 
This is because accurate estimation with the partially Genie-aided PE + OMP algorithm requires higher quantization levels for the LoS channel parameters \( G_{\thetar} \), \( G_{\thetat} \), \( G_{\phir} \), and \( G_{R} \) as the communication distance decreases, due to the increased resolution needed for the near-field channel \cite{ding2023resolution}. 
However, these levels are constrained by the computational complexity of the PE + OMP algorithm, which grows quartically  (i.e. the fourth order) with the parameter quantization level. 
Similarly, the quantization levels $Q_{\eta} \QDr \QDt$ and $\QDr \QDt$ for the 3S-MMV-UOMP and OMP method grow quintically (i.e. the fifth order) and quartically with respect to angular and distance domain samples, respectively. These levels are also constrained by the computational complexity.
On the other hand, the ASAGM method can accommodate much higher parameter quantization levels since its computational complexity grows quadratically with the parameter quantization level.  
Consequently, it guarantees higher estimation performance regardless of the communication distance due to the high resolution of the LoS parameter search grid.
At a distance of 10 meters, the performance of the proposed ASAGM + SMR-OMP method is degraded due to the errors in the SOPM because this distance is smaller than the UPD, which is 12 meters for the given XL-MIMO configuration, as verified in Section \ref{LoS Channel Model}. 
However, the performance loss due to modeling error is negligible over most communication distances that exceed the UPD.

\begin{figure}[t!]
	\centering
	\subfigure[]
	{\includegraphics[width=2.8in]{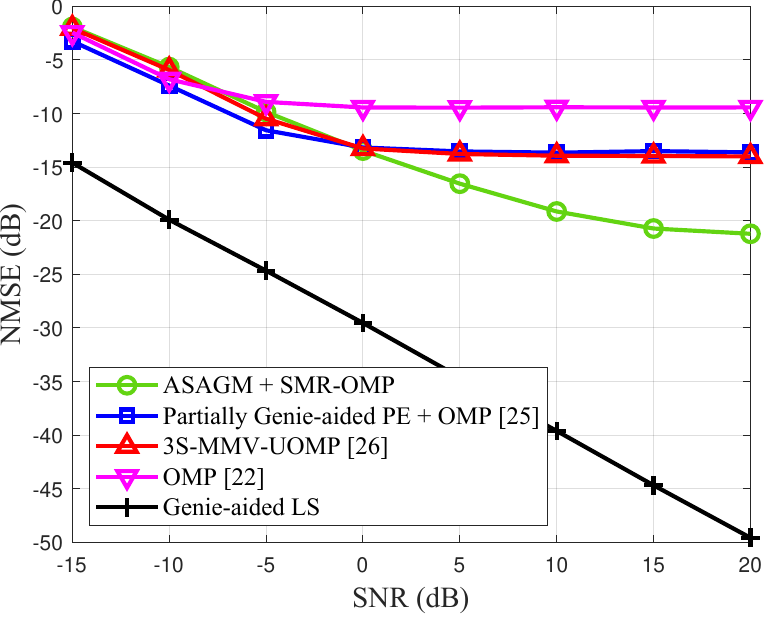}}
	\subfigure[]
	{\includegraphics[width=2.8in]{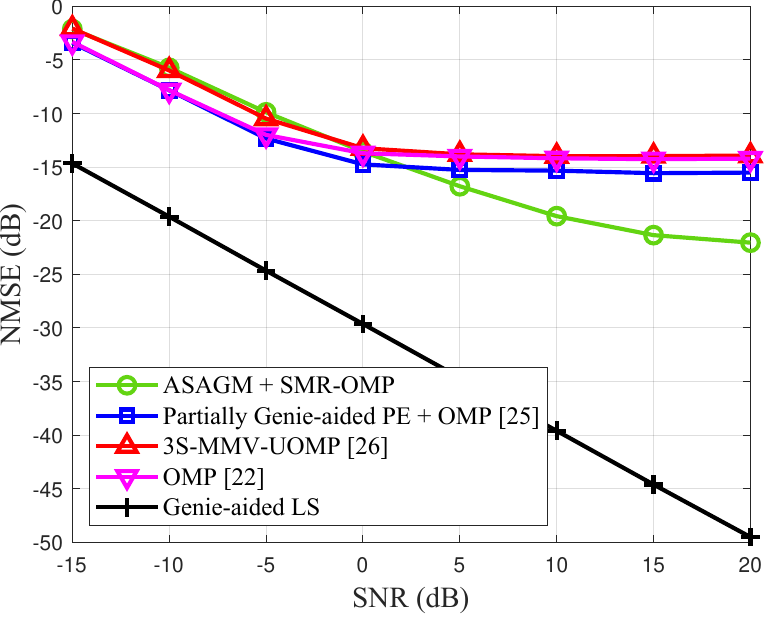}}
	\caption{NMSE versus SNR with the distance between the transmitter and the receiver (a) $R \in [10, 20]$ meters, (b) $R \in [90, 100]$ meters.}
	\label{fig:sim_vs_snr}
	\vspace{-.2in}
\end{figure}

%
Fig. \ref{fig:sim_vs_snr} illustrates the NMSE performance versus SNR when the number of receive and transmit training beams is $\Mr = \Mt = 64$. 
In Fig. \ref{fig:sim_vs_snr} (a) and (b), the distance between the transmitter and the receiver are randomly sampled from $[10, 20]$ meters and $[90, 100]$ meters, respectively.
In both figures, the proposed scheme performs slightly worse than existing benchmark schemes in the low SNR regime (below $0$ dB).
This is because the proposed ASAGM method cannot exploit the SNR gain arising from jointly correlating both receive and transmit array steering vectors, which is useful in the low SNR regime. 
Similarly, the proposed SMR-OMP method cannot leverage this SNR gain in the support detection step since receive and transmit side dictionaries are separately correlated.
However, in the intermediate and high SNR regime (above $0$ dB), the proposed ASAGM + SMR-OMP method outperforms the benchmark schemes. 
This improvement is due to the fact that the parameters constituting transmit or receive array steering vectors can be accurately estimated without relying on the high array gain from joint receive and transmit array gains.
Consequently, the performance is primarily determined by the parameter quantization level, which the proposed ASAGM method can accommodate to a much greater extent with low computational complexity.
However, the performance gap between the proposed method and the Genie-aided LS increases because the proposed scheme is an on-grid method, which inherently has quantization errors.
In addition, the proposed ASAGM + SMR-OMP method performs well consistently regardless of the distance between the transmitter and receiver, while the partially Genie-aided PE + OMP and OMP methods both perform worse when the distance is small at all SNR levels due to their insufficient quantization levels and the increased resolution needed for the near-field channel, as observed in Fig. \ref{fig:sim_vs_distance}.

\begin{figure}[t!]
	\centering
	\subfigure[]
	{\includegraphics[width=2.8in]{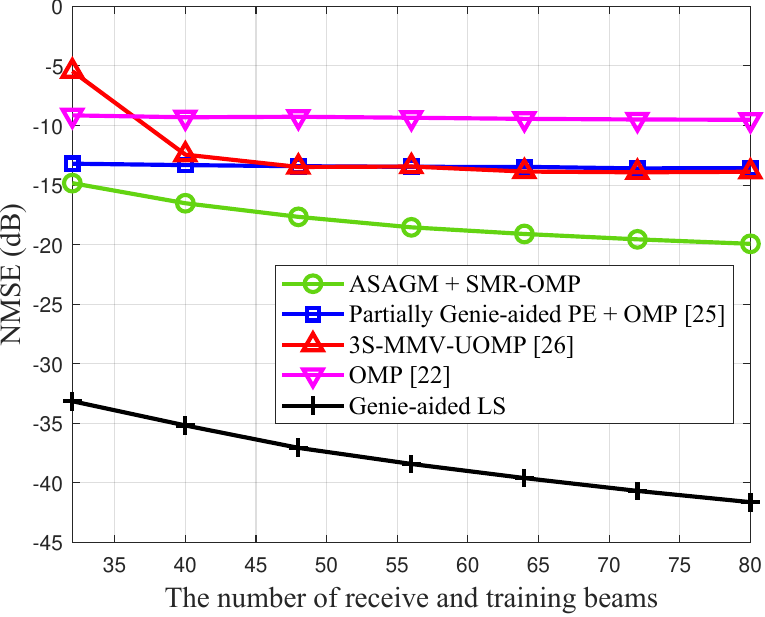}}
	\subfigure[]
	{\includegraphics[width=2.8in]{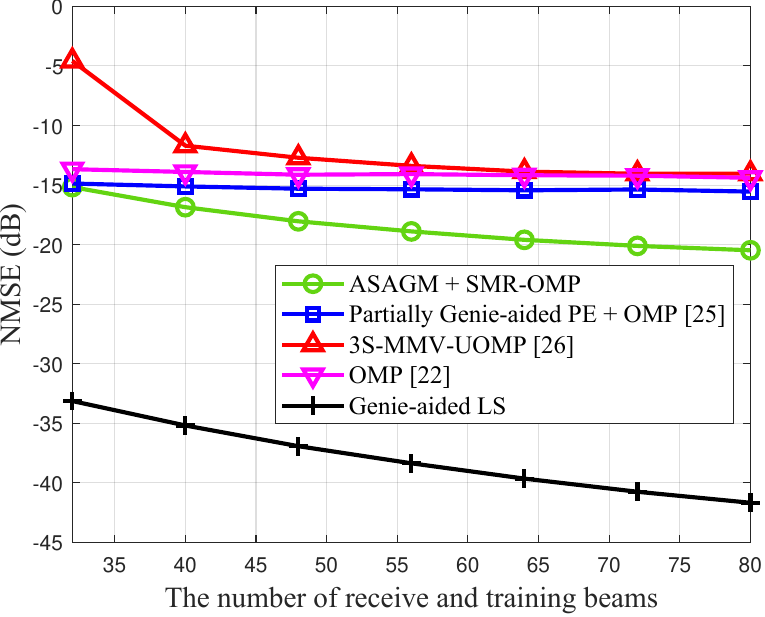}}
	\caption{NMSE versus the number of the receive and transmit training beams ($\Mr = \Mt$) with the distance between the transmitter and receiver (a) $R \in [10, 20]$ meters, (b) $R \in [90, 100]$ meters.}
	\label{fig:sim_vs_pilot}
	\vspace{-0.5cm}
\end{figure}
Fig. \ref{fig:sim_vs_pilot} depicts the NMSE performance versus the number of the receive and transmit training beams where the SNR is $10$ dB.
In Fig. \ref{fig:sim_vs_pilot} (a) and (b), the distance between the transmitter and the receiver are randomly sampled from $[10, 20]$ meters and $[90, 100]$ meters, respectively.
The NMSE of  all  the schemes presented in Fig. \ref{fig:sim_vs_pilot} decreases as the number of training beams increases.
Notably, our proposed ASAGM + SMR-OMP algorithm and the Genie-aided LS algorithm show a significant reduction in NMSE with an increasing number of training beams. 
For ASAGM, this improvement is attributed to the fine resolution of the parameter search grid.
In contrast, the existing methods exhibit minimal NMSE reduction with an increasing number of training beams due to their performance being limited by insufficient quantization levels.

\section{Conclusions}\label{Conclusions}
We addressed the challenge of near-field channel estimation in partially connected XL-MIMO systems. We formulated the near-field channel estimation problem and focused on two primary challenges: accurately estimating the LoS and NLoS components while maintaining low computational complexity. To achieve this, we introduced the SOPM-based LoS channel model, which simplifies complex phase terms into more manageable components with high accuracy. Building on this model, we developed a two-stage channel estimation algorithm: ASAGM for low-complexity LoS channel estimation and SMR-OMP for efficient NLoS channel estimation. Our simulation results demonstrate that the proposed method significantly outperforms existing near-field XL-MIMO channel estimation techniques, particularly in the intermediate and high SNR regime, and in scenarios with arbitrary array placements.  Our future research will focus on exploring channel estimation techniques for systems incorporating extremely large-scale metamaterial-based architectures, such as dynamic metasurface antennas and stacked intelligent metasurfaces.

\appendices
\section{Proof of Lemma 1}\label{Proof of Lemma 1}
To determine the antenna configuration that results in the maximum phase error, we define the Cartesian coordinates for the potential locations of the receive and transmit antenna elements as
\begin{align}
	& (x\ur, y\ur, z\ur) \!=\! (\delta\ur \!\cos{\thetar}, \delta\ur \!\sin{\thetar} \!\sin{\phir}, R \!+\! \delta\ur \!\sin{\thetar} \!\cos{\phir})\!, \hspace{-3pt} \label{coordinate r appendix}  \\
	& (x\ut, y\ut, z\ut) \!=\! (\delta\ut \cos{\thetat}, 0, \delta\ut \sin{\thetat})\!,\label{coordinate t appendix}
\end{align}
where $0 \leq \delta\ut \leq \frac{A\ut}{2}$,  $0 \leq \delta\ur \leq \frac{A\ur}{2}$, and $0 \leq \thetat, \thetar, \phir < 2 \pi$.
For simplicity, we define the relative displacement of the $z$-coordinates as $\Delta_{z} = (z\ur - R) - z\ut$, which satisfies $ \bigl[ (x\ur - x\ut)^2 + y\ur^2 + \Delta_{z}^2 \bigr]^{1/2} \leq \frac{A\ur + A\ut}{2} $.
Then, the distance between the receive and transmit antennas according to the SWM in (\ref{rNUSWM}) and the SOPM in (\ref{rSOPM}) become
\begin{align}
	r & = \bigl[ (x\ur \!-\! x\ut)^2 \!+\! y\ur^2 \!+\! (R \!+\! \Delta_{z})^2 \bigr]^{1/2}, \label{r USWM appendix} \\
	r^{\textrm{SOPM}} & = R + \Delta_{z} \nonumber \label{r SOPM appendix} \\
	& \quad + \! \frac{1}{2R} \Bigl[ \! (x\ur \!-\! x\ut)^2 \!+\! y\ur^2 \!+\! \frac{2 x\ur x\ut}{\delta\ur \delta\ut} (\delta\ur \!-\! \nu\ur) (\delta\ut \!-\! \nu\ut) \! \Bigr]\!,\hspace{-2pt}
\end{align}
where $\nu\ur$ and $\nu\ut$ represent the centroid of the subarrays to which the considered receive and transmit antennas belong.
Next, we define the distance difference function and the constraint function as
\begin{align}
	f(x\ur, x\ut, y\ur, \Delta_{z}) & \triangleq r - r^{\textrm{SOPM}}, \\[-4pt]
	g(x\ur, x\ut, y\ur, \Delta_{z}) & \triangleq (x\ur \!-\! x\ut)^2 + y\ur^2 + \Delta_{z}^2 -\! \left( \! \frac{A\ur \!+\! A\ut}{2} \!\right)^{\!2}\!\!. \hspace{-2pt}
\end{align}
Then, to evaluate the worst case phase error, we solve the phase error maximization problem 
\begin{subequations} \label{problem appendix}
	\begin{align} 
		\max_ { x\ur, x\ut, y\ur, \Delta_{z} } \, &
		\lvert f(x\ur, x\ut, y\ur, \Delta_{z}) \rvert  \\
		\text{s.t. }
		\quad & g(x\ur, x\ut, y\ur, \Delta_{z}) \leq 0.
	\end{align}
\end{subequations}
To address the problem (\ref{problem appendix}), we examine the KKT conditions for the both maximization and minimization of $f(x\ur, x\ut, y\ur, \Delta_{z})$, similar to the approach in \cite{do2023parabolic}.
First, the KKT conditions for the maximization of $f(x\ur, x\ut, y\ur, \Delta_{z})$ are given by
\begin{align}
	\nabla f(x\ur, x\ut, y\ur, \Delta_{z}) + \mu \nabla g(x\ur, x\ut, y\ur, \Delta_{z}) & = \mathbf{0}_{4 \times 1}, \label{KKT 1} \\
	g(x\ur, x\ut, y\ur, \Delta_{z}) & \leq 0, \\
	\mu & \geq 0, \\
	\mu g(x\ur, x\ut, y\ur, \Delta_{z}) & = 0.
\end{align}
Assuming $\mu = 0$, it follows from (\ref{KKT 1}) that $\nabla f(x\ur, x\ut, y\ur, \Delta_{z}) = \mathbf{0}_{4 \times 1}$.
The partial derivative of the difference function with respect to $\Delta_{z}$ is given by
\begin{align}\label{grad1}
& \frac{\partial}{\partial \Delta_{z}}f(x\ur, x\ut, y\ur, \Delta_{z}) \!=\! \frac{R + \Delta_{z}}{ \sqrt{(x\ur \!-\! x\ut)^2 + y\ur^2 + (R \!+\! \Delta_{z})^2} } \!-\! 1.
\end{align}
Equating (\ref{grad1}) to zero gives $x\ur = x\ut$ and $y\ur = 0$.
Letting $x \triangleq x\ur = x\ut$, the partial derivative of the difference function with respect to $x$ is
\begin{align} \label{grad2}
	\frac{\partial}{\partial x}f(x\ur, x\ut, y\ur, \Delta_{z}) & =  \frac{2 x}{R \delta\ur \delta\ut} (\delta\ur -\nu\ur) (\delta\ut - \nu\ut) .
\end{align}
Equating (\ref{grad2}) to zero yields $x = 0$ or $(\delta\ur -\nu\ur) (\delta\ut - \nu\ut) = 0$.
Next, substituting $ x\ur = x\ut = y\ur = 0 $ or $ x\ur = x\ut$, $y\ur = 0$, and $(\delta\ur -\nu\ur) (\delta\ut - \nu\ut) = 0$ into (\ref{r USWM appendix}) and (\ref{r SOPM appendix}), we obtain $ f(x\ur, x\ut, y\ur, \Delta_{z}) = 0$, which is clearly not the global maximum.
Hence, the global maximum occurs when $g(x\ur, x\ut, y\ur, \Delta_{z}) = 0$.
Similarly for the minimization of $f(x\ur, x\ut, y\ur, \Delta_{z})$, the global minimum also occurs when $g(x\ur, x\ut, y\ur, \Delta_{z}) = 0$.
This implies that the necessary condition for maximization of $|f(x\ur, x\ut, y\ur, \Delta_{z})|$ is $g(x\ur, x\ut, y\ur, \Delta_{z}) = 0$, where the shifted receive antenna coordinate $ (x\ur, y\ur, z\ur - R) $ and $ (x\ut, y\ut, z\ut) $ are symmetric with respect to $(0, 0, 0)$, and placed as far from the origin as possible. 
This corresponds to $\delta\ur = \frac{A\ur}{2} $, $\delta\ut = \frac{A\ut}{2} $, $\thetat = \thetar + \pi$, and $\phir = 0$. Next, since the optimization variables that satisfies the necessary condition for optimality can be parameterized by $\theta \triangleq \thetat = \thetar + \pi$, the optimal solution is obtained by determining the value of $\theta$ that maximizes $|f(x\ur, x\ut, y\ur, \Delta_{z})|$.
To this end, the coordinates that satisfies the necessary condition for optimality of (\ref{problem appendix}) can be expressed using (\ref{coordinate r appendix}) and (\ref{coordinate t appendix}) as
\begin{align}
	& (x\ur^{\star}, x\ut^{\star}, y\ur^{\star}, \Delta_{z}^{\star}) \!=\! \left(\!-\frac{A\ur}{2} \!\cos{\theta}, \frac{A\ut}{2} \!\cos{\theta}, 0, -\frac{A\ur \!+\! A\ut}{2} \!\sin\theta \!\right)\!\!.
\end{align}
Then, (\ref{r USWM appendix}) and (\ref{r SOPM appendix}) can be rewritten as
\begin{align}
	r & = R \bigl( 1 - 2 p \sin{\theta}  + p^2 \bigr)^{1/2}, \label{r USWM appendix 2} \\[-3pt]
	r^{\textrm{SOPM}} & = R \bigl( 1 - p \sin{\theta}  + \frac{p^2-q^2}{2} \cos^2{\theta} \bigr), \label{r SOPM appendix 2}
\end{align}
where $p = \frac{\Ar + \At}{2R}$ and $q = \sqrt{\frac{ (\Ar - 2 \nu\ur) (\At - 2 \nu\ut) }{2R^2}}$.
Subsequently, to determine $\theta$ that maximizes the phase error in a tractable way, we rewrite (\ref{r USWM appendix 2}) in a power series form based on the Taylor series expansion as
\begin{align} \label{r USWM appendix 3}
	r = R \left(1 - p \sin \theta + \frac{p^2 \cos^2 \theta}{2} \right) + \mathcal{O}(R p^4) .
\end{align}
Hence, from (\ref{r SOPM appendix 2}) and (\ref{r USWM appendix 3}), we obtain the phase discrepancy 
\begin{align}\label{phase discrepancy appendix}
	\lvert r - r^{\textrm{SOPM}} \rvert & = \frac{R q^2 \cos^2 \theta}{2} + \mathcal{O}(R p^4) \nonumber \\
	& = \frac{(\Ar - 2 \nu\ur)(\At - 2 \nu\ut) \cos^2 \theta}{4R} + \mathcal{O}(R p^4) \nonumber \\
	& \overset{(a)}{=} \frac{(\frac{\Nr}{\Kr} - 1) (\frac{\Nt}{\Kt} - 1) d^2 \cos^2 \theta }{4R} + \mathcal{O}(R p^4) \nonumber \\
	& = \frac{\Ars \Ats \cos^2 \theta}{4R}+ \mathcal{O}(R p^4),
\end{align}
where (a) holds because $\nu\ur = \frac{1}{2}(1 - \frac{1}{\Kr})Nd$ and $\nu\ut = \frac{1}{2}(1 - \frac{1}{\Kt})Nd$, as derived from (\ref{coupled term approximation}). This is due to the fact that the receive and transmit antenna locations, which solve problem (\ref{problem appendix}), are positioned in the subarray furthest from the origin.
Additionally, $ \Ars = (\Nrs - 1)d $ and $ \Ats = (\Nts - 1)d $ represent the aperture of the receive and the transmit subarray, respectively.
Finally, by maximizing (\ref{phase discrepancy appendix}) over $\theta$, the worst case phase error is obtained as
\begin{align}
	\max_{\theta} \frac{2 \pi}{\lambda} \lvert r - r^{\textrm{SOPM}} \rvert & = \frac{\pi \Ars \Ats}{2R \lambda}+ \mathcal{O}(R p^4).
\end{align}

%
%

\ifCLASSOPTIONcaptionsoff
\newpage
\fi

\bibliographystyle{IEEEtran}
\bibliography{IEEEabrv,references}

\begin{thebibliography}{10}
\providecommand{\url}[1]{#1}
\csname url@samestyle\endcsname
\providecommand{\newblock}{\relax}
\providecommand{\bibinfo}[2]{#2}
\providecommand{\BIBentrySTDinterwordspacing}{\spaceskip=0pt\relax}
\providecommand{\BIBentryALTinterwordstretchfactor}{4}
\providecommand{\BIBentryALTinterwordspacing}{\spaceskip=\fontdimen2\font plus
\BIBentryALTinterwordstretchfactor\fontdimen3\font minus \fontdimen4\font\relax}
\providecommand{\BIBforeignlanguage}[2]{{%
\expandafter\ifx\csname l@#1\endcsname\relax
\typeout{** WARNING: IEEEtran.bst: No hyphenation pattern has been}%
\typeout{** loaded for the language `#1'. Using the pattern for}%
\typeout{** the default language instead.}%
\else
\language=\csname l@#1\endcsname
\fi
#2}}
\providecommand{\BIBdecl}{\relax}
\BIBdecl

\bibitem{andrews2014will}
J.~G. Andrews, S.~Buzzi, W.~Choi, S.~V. Hanly, A.~Lozano, A.~C. Soong, and J.~C. Zhang, ``What will 5{G} be?'' \emph{IEEE J. Sel. Areas Commun.}, vol.~32, no.~6, pp. 1065--1082, 2014.

\bibitem{viswanathan2020communications}
H.~Viswanathan and P.~E. Mogensen, ``Communications in the 6{G} era,'' \emph{IEEE Access}, vol.~8, pp. 57\,063--57\,074, 2020.

\bibitem{wang2023extremely}
Z.~Wang, J.~Zhang, H.~Du, W.~E.~I. Sha, B.~Ai, D.~Niyato, and M.~Debbah, ``Extremely large-scale {MIMO}: Fundamentals, challenges, solutions, and future directions,'' \emph{IEEE Wireless Commun.}, vol.~31, no.~3, pp. 117--124, 2024.

\bibitem{tataria20216g}
H.~Tataria, M.~Shafi, A.~F. Molisch, M.~Dohler, H.~Sjöland, and F.~Tufvesson, ``6{G} wireless systems: Vision, requirements, challenges, insights, and opportunities,'' \emph{Proc. IEEE}, vol. 109, no.~7, pp. 1166--1199, 2021.

\bibitem{alkhateeb2014channel}
A.~Alkhateeb, O.~El~Ayach, G.~Leus, and R.~W. Heath, ``Channel estimation and hybrid precoding for millimeter wave cellular systems,'' \emph{IEEE J. Sel. Top. Signal Process.}, vol.~8, no.~5, pp. 831--846, 2014.

\bibitem{bogale2016number}
T.~E. Bogale, L.~B. Le, A.~Haghighat, and L.~Vandendorpe, ``On the number of {RF} chains and phase shifters, and scheduling design with hybrid analog–digital beamforming,'' \emph{IEEE Trans. Wireless Commun.}, vol.~15, no.~5, pp. 3311--3326, 2016.

\bibitem{gao2016energy}
X.~Gao, L.~Dai, S.~Han, C.-L. I, and R.~W. Heath, ``Energy-efficient hybrid analog and digital precoding for mmwave {MIMO} systems with large antenna arrays,'' \emph{IEEE J. Sel. Areas Commun.}, vol.~34, no.~4, pp. 998--1009, 2016.

\bibitem{majidzadeh2017hybrid}
M.~Majidzadeh, A.~Moilanen, N.~Tervo, H.~Pennanen, A.~Tölli, and M.~Latva-aho, ``Hybrid beamforming for single-user {MIMO} with partially connected {RF} architecture,'' in \emph{Proc. Eur. Conf. Netw. Commun. (EuCNC)}, 2017, pp. 1--6.

\bibitem{du2018hybrid}
J.~Du, W.~Xu, H.~Shen, X.~Dong, and C.~Zhao, ``Hybrid precoding architecture for massive multiuser {MIMO} with dissipation: Sub-connected or fully connected structures?'' \emph{IEEE Trans. Wireless Commun.}, vol.~17, no.~8, pp. 5465--5479, 2018.

\bibitem{song2019fully}
X.~Song, T.~Kühne, and G.~Caire, ``Fully-/partially-connected hybrid beamforming architectures for mmwave {MU-MIMO},'' \emph{IEEE Trans. Wireless Commun.}, vol.~19, no.~3, pp. 1754--1769, 2020.

\bibitem{lee2016channel}
J.~Lee, G.-T. Gil, and Y.~H. Lee, ``Channel estimation via orthogonal matching pursuit for hybrid {MIMO} systems in millimeter wave communications,'' \emph{IEEE Trans. Commun.}, vol.~64, no.~6, pp. 2370--2386, 2016.

\bibitem{rodriguez2018frequency}
J.~Rodríguez-Fernández, N.~González-Prelcic, K.~Venugopal, and R.~W. Heath, ``Frequency-domain compressive channel estimation for frequency-selective hybrid millimeter wave {MIMO} systems,'' \emph{IEEE Trans. Wireless Commun.}, vol.~17, no.~5, pp. 2946--2960, 2018.

\bibitem{zhou2015spherical}
Z.~Zhou, X.~Gao, J.~Fang, and Z.~Chen, ``Spherical wave channel and analysis for large linear array in {L}o{S} conditions,'' in \emph{Proc. IEEE Globecom Workshops (GC Wkshps)}, 2015, pp. 1--6.

\bibitem{le2019massive}
L.~L. Magoarou, A.~L. Calvez, and S.~Paquelet, ``Massive {MIMO} channel estimation taking into account spherical waves,'' in \emph{Proc. IEEE 20th Int. Workshop Signal Process. Adv. Wireless Commun. (SPAWC)}, 2019, pp. 1--5.

\bibitem{goodman2005introduction}
J.~W. Goodman, \emph{Introduction to Fourier optics}.\hskip 1em plus 0.5em minus 0.4em\relax Roberts and Company publishers, 2005.

\bibitem{bohagen2005construction}
F.~Bohagen, P.~Orten, and G.~Oien, ``Construction and capacity analysis of high-rank line-of-sight {MIMO} channels,'' in \emph{Proc. IEEE Wireless Commun. Netw. Conf.}, vol.~1, 2005, pp. 432--437 Vol. 1.

\bibitem{do2020reconfigurable}
H.~Do, N.~Lee, and A.~Lozano, ``Reconfigurable {ULA}s for line-of-sight {MIMO} transmission,'' \emph{IEEE Trans. Wireless Commun.}, vol.~20, no.~5, pp. 2933--2947, 2021.

\bibitem{do2023parabolic}
------, ``Parabolic wavefront model for line-of-sight {MIMO} channels,'' \emph{IEEE Trans. Wireless Commun.}, vol.~22, no.~11, pp. 7620--7634, 2023.

\bibitem{xi2023gridless}
Y.~Xi, F.~Zhu, B.~Zhou, T.~Liu, and S.~Ma, ``Gridless hybrid-field channel estimation for extra-large aperture array massive {MIMO} systems,'' \emph{IEEE Wireless Commun. Lett.}, vol.~13, no.~2, pp. 496--500, 2024.

\bibitem{selvan2017fraunhofer}
K.~T. Selvan and R.~Janaswamy, ``Fraunhofer and {F}resnel distances: Unified derivation for aperture antennas,'' \emph{IEEE Antennas Propag. Mag.}, vol.~59, no.~4, pp. 12--15, 2017.

\bibitem{han2020channel}
Y.~Han, S.~Jin, C.-K. Wen, and X.~Ma, ``Channel estimation for extremely large-scale massive {MIMO} systems,'' \emph{IEEE Wireless Commun. Lett.}, vol.~9, no.~5, pp. 633--637, 2020.

\bibitem{cui2022channel}
M.~Cui and L.~Dai, ``Channel estimation for extremely large-scale {MIMO}: Far-field or near-field?'' \emph{IEEE Trans. Commun.}, vol.~70, no.~4, pp. 2663--2677, 2022.

\bibitem{zhang2023near}
X.~Zhang, H.~Zhang, and Y.~C. Eldar, ``Near-field sparse channel representation and estimation in 6{G} wireless communications,'' \emph{IEEE Trans. Commun.}, vol.~72, no.~1, pp. 450--464, 2024.

\bibitem{kang2024pilot}
Y.~Kang, H.~Seo, and W.~Choi, ``Pilot signal and channel estimator co-design for hybrid-field {XL-MIMO},'' \emph{IEEE Trans. Commun.}, 2024, Early access.

\bibitem{lu2023near}
Y.~Lu and L.~Dai, ``Near-field channel estimation in mixed {L}o{S}/{NL}o{S} environments for extremely large-scale {MIMO} systems,'' \emph{IEEE Trans. Commun.}, vol.~71, no.~6, pp. 3694--3707, 2023.

\bibitem{shi2024double}
X.~Shi, J.~Wang, X.~Wang, C.~You, and J.~Song, ``Double-sided near-field {XL-MIMO}: Beamfocusing codeword selection and channel estimation,'' \emph{IEEE Trans. Commun., Early access}, 2024.

\bibitem{tarboush2024cross}
S.~Tarboush, A.~Ali, and T.~Y. Al-Naffouri, ``Cross-field channel estimation for ultra massive-{MIMO} {TH}z systems,'' \emph{IEEE Trans. Wireless Commun.}, vol.~23, no.~8, pp. 8619--8635, 2024.

\bibitem{lu2021communicating}
H.~Lu and Y.~Zeng, ``Communicating with extremely large-scale array/surface: Unified modeling and performance analysis,'' \emph{IEEE Trans. Wireless Commun.}, vol.~21, no.~6, pp. 4039--4053, 2022.

\bibitem{zhou2017low}
Z.~Zhou, J.~Fang, L.~Yang, H.~Li, Z.~Chen, and R.~S. Blum, ``Low-rank tensor decomposition-aided channel estimation for millimeter wave {MIMO}-{OFDM} systems,'' \emph{IEEE J. Sel. Areas Commun.}, vol.~35, no.~7, pp. 1524--1538, 2017.

\bibitem{tropp2005simultaneous}
J.~Tropp, A.~Gilbert, and M.~Strauss, ``Simultaneous sparse approximation via greedy pursuit,'' in \emph{Proc IEEE Int. Conf. Acoustics, Speech, Signal Process. (ICASSP)}, vol.~5, 2005, pp. v/721--v/724 Vol. 5.

\bibitem{ding2023resolution}
Z.~Ding, ``Resolution of near-field beamforming and its impact on {NOMA},'' \emph{IEEE Wireless Commun. Lett.}, vol.~13, no.~2, pp. 456--460, 2024.

\end{thebibliography}

\end{document}